\documentclass[letter,aps,pra,showpacs,twocolumn,superscriptaddress]{revtex4-2}   
   
\usepackage[utf8]{inputenc}  
\usepackage[T1]{fontenc}     
\usepackage[british]{babel}  
\usepackage{lmodern}
\usepackage{url}

\usepackage[scaled=1.03]{inconsolata} 
\usepackage[usenames,dvipsnames]{color} 
\usepackage[colorlinks,citecolor=blue,linkcolor=magenta,urlcolor=blue]{hyperref}  
\usepackage{graphicx} 
\usepackage{tikz}
\usepackage[babel]{microtype}  
\usepackage{amsmath,amssymb,amsthm,bm,mathtools,amsfonts,mathrsfs,bbm,dsfont} 
\usepackage{xspace}  
\usepackage{multirow}
\usepackage{subcaption}
\usepackage{verbatim}
\usepackage{enumerate}
\usepackage{tcolorbox}

\usepackage{physics}

\usepackage{bbold}
\usepackage[mathscr]{eucal}

\newcommand{\bs}[1]{\boldsymbol{#1}}

\newcommand{\mrm}[1]{\mathrm{#1}}

\newtheoremstyle{mystyle}
  {6pt}
  {6pt}
  {\normalfont}
  {0pt}
  {\bf}
  {.}
  { }
  {}

\theoremstyle{mystyle}
\newtheorem{theorem}{Theorem}
\newtheorem{lemma}{Lemma}

\newtheorem{corollary}{Corollary}

\newcommand{\Tom}[1]{\textcolor{green}{#1}}

\begin{document}

\title{Quantifying the high-dimensionality of quantum devices}
\date{\today}

\author{Thomas Cope}
\affiliation{Institut für Theoretische Physik, Leibniz Universität Hannover,
30167 Hannover, Germany}
\author{Roope Uola}
\affiliation{Department of Applied Physics, University of Geneva, 1211 Geneva, Switzerland}

\begin{abstract}
    We introduce a measure of average dimensionality (or coherence) for high-dimensional quantum devices. This includes sets of quantum measurements, steering assemblages, and quantum channels. For measurements and channels, our measure corresponds to an average compression dimension, whereas for quantum steering we get a semi-device independent quantifier for the average entanglement dimensionality known as the Schmidt measure. We analyze the measure in all three scenarios. First, we show that it can be decided via semi-definite programming for channels and measurements in low-dimensional systems. Second, we argue that the resulting steering measure is a high-dimensional generalisation of the well-known steering weight. Finally, we analyse the behavior of the measure in the asymptotic setting. More precisely, we show that the asymptotic Schmidt measure of bipartite quantum states is equal to the entanglement cost and show how the recently introduced entanglement of formation of steering assemblages can be related to our measure in the asymptotic case.
\end{abstract}

\maketitle

\section{Introduction}

The advent of quantum information processing has led to various improvements in our understanding of the quantum world, which again has provided innovations reaching even commercial applications, such as devices capable of performing quantum key distribution and randomness generation \cite{GisinCrypto,ReviewQKDStefano,ReviewQKD,ReviewRNG,ReviewRNGMa}. Many of these developments have concentrated on possibly low-dimensional quantum systems consisting of many qubits \cite{ReviewcomputingphotonicsKok,Reviewsemicondqubits,Reviewmultiphoton}, or on degrees of freedom that are given by continuous variables, such as position and momentum \cite{CVreview}.

Now, with the second quantum revolution at hand, one particularly noticeable direction of research has been the development of noise-resilient and loss-tolerant quantum communication protocols with high quantum information carrying capacity, see \cite{Krenn14,Fickler16,HDBellErika,Malik16,HDreviewcomputationWang}. One way to reach this regime is to shift from systems having two degrees of freedom, such as spins of electrons or polarization of photons, to high-dimensional degrees of freedom given by, e.g., time-bins or orbital angular momentum \cite{OAMreview,HDsystreviewErhard}. Various known theoretical tools can be directly applied to these platforms, including entanglement witnessing techniques \cite{ReviewEntanglementOtfried}, numerical optimization methods \cite{cavalcanti2016quantum}, and general quantification strategies for quantum resources \cite{Takagi19,uola2019quantifying}.

When concentrating on, e.g., quantum communication protocols using high-dimensional systems, it is important that one possesses genuinely high-dimensional resources, in order to reach the aforementioned advantages. As an example, in a prepare-and-measure scenario one should verify that one’s quantum devices, i.e. the state preparation, the quantum channel and the quantum measurement, do not reduce the system into an effectively lower-dimensional one, hence, compromising e.g. the high noise-tolerance or information carrying capacity. In other words, the devices should not allow for a simulation strategy using only lower-dimensional devices. This problem has received some attention from the theory perspective in the case of entanglement theory \cite{TH2000,EB2001}, but it has been only more recently investigated from the perspective of quantum channels \cite{CK2006,shirokov2011schmidt}, quantum measurements \cite{ioannou2022simulability}, and correlations stronger than entanglement \cite{Detal2021}.

In this manuscript, we develop a resource measure for quantum channels, quantum measurements, and Einstein-Podolsky-Rosen steering. The main point of our measure is that it enables one to rule out low-dimensional simulation models for each of these classes of quantum devices. The measure is based on the Schmidt measure of entanglement \cite{EB2001}, which can be seen as a high-dimensional version of the celebrated best separable approximation \cite{LewensteinBestSeparable}.

We develop a Schmidt-type measure for the three types of devices mentioned above. This results in quantifiers for the average dimensionality needed to realise a given quantum device. This is in contrast to known dimensionality quantifiers for channels \cite{CK2006}, measurements \cite{ioannou2022simulability} and state assemblages \cite{Detal2021}, which quantify the highest required dimension. To give a simple example on what we mean by average and highest required dimensions, one can think of the two-qubit state $\lambda|\psi^+\rangle\langle\psi^+|+(1-\lambda)\openone/4$. The Schmidt number, i.e. the number of degrees of freedom that one needs to be able to entangle to produce the state, is $\log(2)=1$ if the state is entangled, i.e. when $1/3<\lambda\leq 1$, and $\log(1)=0$ otherwise. However, the Schmidt number does not contain more precise information about the strength of the entanglement, i.e. about the parameter $\lambda$. One way around this is to use the Schmidt measure instead. The Schmidt measure quantifies the average dimensionality needed to prepare the state, i.e. in this example it counts the relative frequency between the uses of entangled and separable preparations. For the example state, this number turns out to be \Tom{$(3\lambda-1)/2$ for $\lambda\geq 1/3$, 0 otherwise}.  One can replace the two-qubit state here by a high-dimensional one, for which, even when being noisy, the Schmidt number can be maximal \cite{TH2000}, i.e. equal to $\log(d)$, while the Schmidt measure remains below this value. In contrast to the existing literature on dimensionality of measurements and state assemblages, where the main focus has been on evaluating \cite{ioannou2022simulability,Des21,CarlosHDsteerhierarchy} and experimentally testing \cite{Detal2021} the highest required dimension, we concentrate on formulating the average measures for the mentioned types of devices, investigating their properties and their connections to known quantities in, e.g., entanglement theory. First, we show that the measure can be efficiently decided by semi-definite programming (SDP for short) techniques in the low-dimensional scenario of qubit-to-qutrit as well as qutrit-to-qubit channels. Second, we discuss how the measure presents a quantifier of average dimensionality of quantum measurements through an explicit simulation protocol, and provide an alternative formulation of the measure in this scenario. Third, in the case of quantum steering, we argue that our measure represents a semi-device independent verification method for the Schmidt measure of bipartite quantum states and show that it generalises the known resource measure given by the steering weight \cite{Steeringweight}. Finally, we analyze the behavior of our measure in the asymptotic setting. We prove that the asymptotic Schmidt measure of bipartite states equals the entanglement cost, hence, generalizing the results of \cite{Buscemi11} concerning the Schmidt number, and analyze how such result generalizes to the semi-device independent setting by using a recently introduced entanglement of formation for steering assemblages \cite{C2021}.

\section{Dimensionality of a Quantum Channel}
In this section, we ask the question: given a channel $\mathcal{E}:L(\mathcal{H}_{d_1})\rightarrow L(\mathcal{H}_{d_2})$ with $d_1$ and $d_2$ finite, what is the dimension of the channel? To answer this, we first should be more precise about the question. We would typically associate $L(\mathcal{H}_{d_1})$ with dimension $d_1$ and $L(\mathcal{H}_{d_2})$ with dimension $d_2$; we will refer to these as the ``input/output" dimensions of the channel respectively. What is more subtle, and what we are interested in, is the level of coherence preserved by the channel. Consider a noiseless channel $id_{100}:L(\mathcal{H}_{100})\rightarrow L(\mathcal{H}_{100})$. This is capable of preserving coherence in the entire $d=100$ space. However, the channel $\mathcal{E}:L(\mathcal{H}_{100})\rightarrow L(\mathcal{H}_{100})$, $\mathcal{E}(\rho)=\rho_{\mathrm{diag}}$, destroys all coherence, despite acting between the same two spaces. We would like to characterise $\mathcal{E}$ as dimensionality $0$, and $id_{100}$ as dimensionality $100$ (technically, as we shall see, $\log 100$). Note that this is not the same as the \emph{capacity} of the channel; we require that truly $d$ dimensional coherence is preserved, but not arbitrary $d$ dimensional states.

One can think of an analogous question for quantum states; a pure quantum state can be written in its Schmidt decomposition \cite{S1907}, $\ket{\Phi} = \sum^{d}_{i} \lambda_{i} \ket{e_{i}f_{i}},\;\lambda_i>0, \sum_{i}\lambda_{i}^2=1$. To create such a state, we require a genuinely entangled space of dimension $d=:S_R(\Phi)$, its \emph{Schmidt rank}. This can be extended to mixed states, by decomposing them to pure states, and then checking the largest Schmidt rank required. This defines the \emph{Schmidt number} \cite{TH2000},
$S_N(\rho):=\min_{p_i,\Phi_i} \max_{i} \log S_R(\Phi_i)$, such that $\sum_{i} p_{i}\ket{\Phi_{i}}\bra{\Phi_{i}} = \rho$, and $p_i$ is a valid probability distribution (which we will not state explicitly throughout the rest of the paper). The logarithm appears such that the Schmidt number is additive over the tensor product. Intuitively, the Schmidt number gives the number of entangled degrees of freedom required to create the (mixed) state.

How does this translate into quantum channels? We may always write a quantum channel using an explicit Kraus representation, $\mathcal{E}(\rho)=\sum_{i}K_{i}\rho K^{\dagger}_{i}$, $\sum_{i}K_{i}^{\dagger}K_{i}=\openone_{d_1}$, where $\openone_{d_1}$ is the identity matrix acting on a $d_1$-dimensional space. If all of the (matrix) ranks of the Kraus operations are less than some value $d$, then the channel cannot preserve coherence greater than $d$, as the kernel of each matrix is simply too large.

This leads to a definition of a Schmidt number for channels in the following way: $S_N(\mathcal{E})=\min_{K_i} \max_{i} \log \mathrm{rank}(K_i)$, where we minimise over all Kraus representations of the channel \cite{CK2006}. This definition is reinforced by the result that this Schmidt number of a channel $\mathcal{E}$ and the Schmidt number of the corresponding Choi matrix $\chi_{\mathcal{E}}$ coincide \cite{CK2006}. One observation made in \cite{CK2006} is that a channel with Schmidt number $n$, when applied to one half of a bipartite system, cannot preserve a Schmidt rank of a state greater than $n$. Thus these channels are known as $n$-\emph{partially entanglement breaking} (n-PEB) channels \cite{CK2006}.\\

Instead of asking what size of entangled subspace is needed to create a given mixed state, we could instead ask what is the \emph{average} size needed. This is given by the Schmidt measure \cite{EB2001}, $S_M(\rho)=\min_{p_i,\Phi_i} \sum_{i} p_i \log S_R(\Phi_i)$, such that $\sum_{i} p_{i}\ket{\Phi_{i}}\bra{\Phi_{i}} = \rho$. In Appendix \ref{One_Shot_Proof}, we show that this is the expected one-shot cost of preparing $\rho$, provided one is allowed to begin with a labelled mixture of maximally entangled states. We can again consider the same question for channels - is there a notion of the \emph{average} coherence preserved? One can directly translate the Schmidt measure for states to the Choi matrix of the channel - as we shall see later, this would not result in the operational meaning we desire. Instead, we should weight each coherent output system by the probability it is obtained. This itself is dependent on the state being inputted. Since we do not call a channel entanglement breaking as long as there is some input state whose output is entangled, we should also give the channel the best possible state to exhibit its dimensionality. Thus we arrive at the \emph{dimension measure} of a quantum channel,
\begin{equation}
D_M(\mathcal{E}):=\min_{K_i} \max_{\rho}\sum_i\mathrm{Tr}(K_i^{\dagger}K_i\rho) \log \mathrm{rank}(K_i).
\end{equation}
Operationally, we argue $\rho$ should \emph{not} be included in the $\log \mathrm{rank}$. The intuition is that the $K_i$ represent a realisation (chosen by us) of the channel, into which is inputted an arbitrary quantum state $\rho$ which we have no knowledge about.\\

The above defined quantity has a number of desirable properties:
\begin{itemize}
\item It is convex (under combination of channels).
\item It is subadditive over tensor product.
\item It is non-increasing under quantum pre-processing and quantum post-processing.
\item Entanglement breaking channels have dimension measure 0.
\item Isometric channels have dimension measure $\log d_1$.
\end{itemize}
Additionally, this measure is able to capture noise properties of the channel that simply looking at the Schmidt number cannot. Figure 1 compares the Schmidt number and dimension measure of the qubit depolarising channel, $\mathcal{E}_p(\rho)= (1-p)\rho + p\openone/2$. Rather than a discontinuous jump, the dimension measure decreases linearly with noise.
\begin{figure}
    \centering
    \includegraphics[width=\columnwidth]{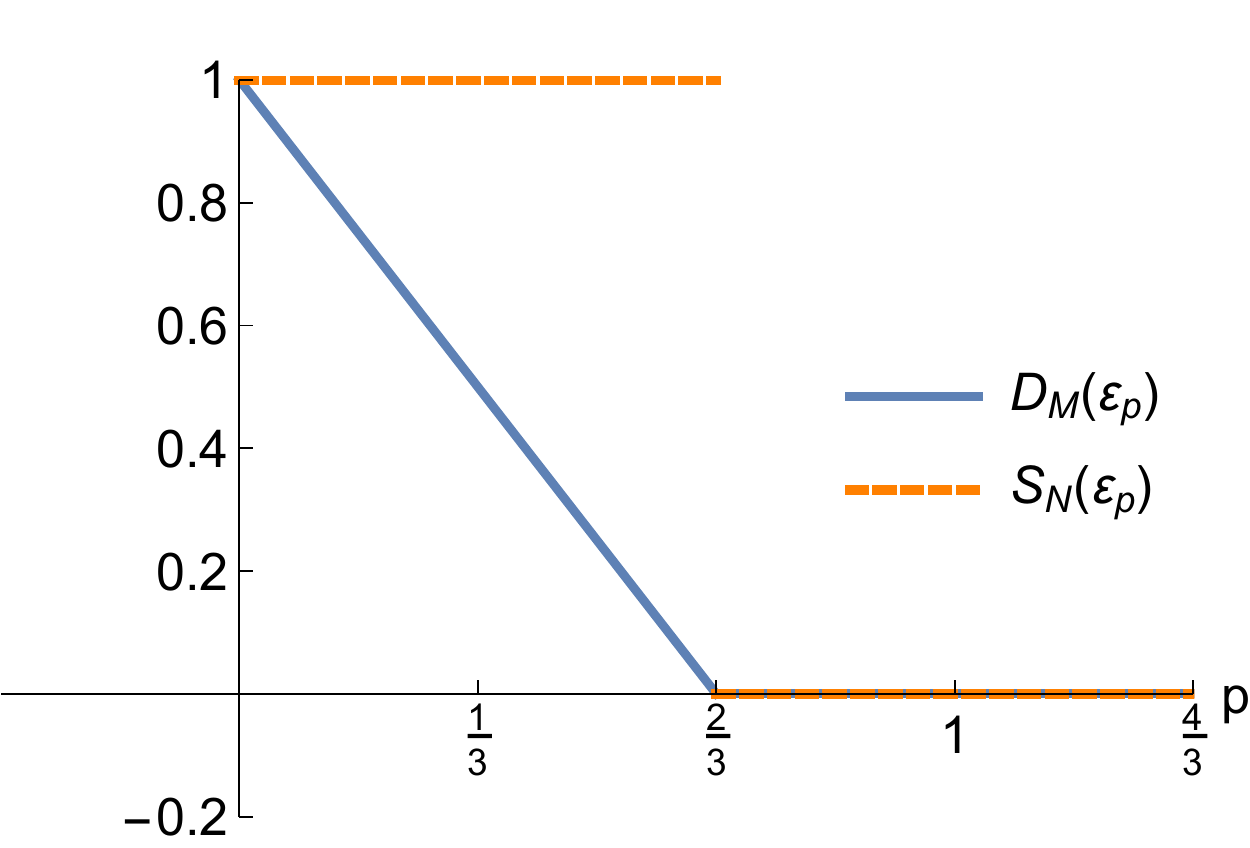}
    \caption{A comparison of the dimension measure ($D_M[\mathcal{E}_p]$) and the Schmidt number ($S_N[\mathcal{E}_p]$) for the depolarising channel. The increased noise is reflected in a lower dimension measure.}
    \label{Depol_Channel}
\end{figure}
We will see later on that the dimension measure is not always continuous; rather it satisfies a weaker property known as \emph{lower semi-continuity}. Intuitively, this means that one can always find a small enough deviation for which the measure does not drastically decrease. Since the rank of an operator is itself a lower semi-continuous function, we should not expect a stronger level of continuity than this. One could also argue the definition should have an infimum and supremum, especially since the set of Kraus representations of a channel is not compact; however, this definition turns out to be equal to the one above. This is mainly a result of the finite dimension $d_1$, which tells us the measure can be realised with $d_1^2$ summands i.e. Kraus operators. These properties: lower semi-continuity and achievable with a finite sum, are properties shared by all the subsequent measures we will introduce in this paper. 

Although we cannot interpret the dimension measure as a property of the Choi matrix of the channel in such a neat way as we can for the Schmidt number, we can nevertheless connect it to the Schmidt measure of entangled states sent through the channel. We have the following relation: $\sup_{\ket{\phi}} S_{M}[id\otimes\mathcal{E}\left(\ket{\phi}\bra{\phi}\right)]\leq D_{M}(\mathcal{E})$. The choice of entangled input state $\phi$ corresponds to choosing $\rho$, and the specific decomposition to a specific set of Kraus operators. The reason for the inequality is that for the two quantities the order of supremum/minimisation are reversed. It may be with a suitable minimax theorem this relation could be made an equality.

\section{SDP for dimensionality of a quantum channel}
For a channel $\mathcal{E}:\mathcal{L}(\mathcal{H}_{d_1})\rightarrow \mathcal{L}(\mathcal{H}_{d_2})$, we wish to calculate the following quantity: $\min_{K_i} \max_{\rho} \sum_i \mathrm{Tr}(K_i^{\dagger}K_i\rho)\mathrm{log}\;\mathrm{rank}(K_i)$. Via the Choi isomorphism, we can equivalently write this as an optimisation over the Choi matrix
\begin{equation}\label{ChoiDecomp}
\min_{\phi^l} \max_\rho \sum_l d\,\mathrm{Tr}(\rho^{T}\otimes \openone\ket{\phi^l}\bra{\phi^l})\log S_R(\ket{\phi^l}),
\end{equation}
where we can associate $\phi_{l}=\sum_{a,b} [K_l]_{a,b}/\sqrt{d}\ket{ba}$, and require that $\sum_{l} \ket{\phi_l}\bra{\phi_l} = \chi_{\mathcal{E}}$.

In the case where our channel has at least one system of dimension 2, then this equation simplifies to (a detailed derivation is given in Appendix \ref{Eq2to3}):
\begin{align}\label{SR2}
&\min_{\sigma^0\in \mathrm{SEP}} \max_\rho d\,\mathrm{Tr}(\rho^{T}\otimes\openone\left(\chi_{\mathcal{E}} -\sigma^0\right))\\
\equiv &\min_{\sigma^0\in \mathrm{SEP}} \max_\rho \mathrm{Tr}\left(\rho^{T}\left(\openone-d\,\sigma^0_{A}\right)\right)\nonumber
\end{align}
where $\sigma^0_{A}=\mathrm{Tr}_{B}\sigma^0$. Noting that transposition is a positive trace-preserving channel, we can see that our maximisation is equivalent to taking the largest eigenvalue of $(\openone-d\,\sigma^0_{A})$. Given a specific $\sigma^0$, this can be formulated as an SDP. Furthermore, for channels where $d_1d_2\leq 6$, then the separable states are equivalent to states with a partial positive transpose \cite{P1996,HHH1996}. This allows the above optimisation to be formulated via an SDP:

\begin{align}
&\text{minimize}\; \alpha,\\
&\text{subject to:}\nonumber\\
&\chi_{\mathcal{E}}-\sigma^0\geq 0,\nonumber\\
&(\sigma^0)^{T_A}\geq 0,\nonumber\\
&(\alpha-1)\openone+d\sigma^0_{A}\geq 0.\nonumber
\end{align}

By looking at this SDP analytically, we can see that for certain channels, the Schmidt number and dimension measure coincide, meaning the dimension measure is not continuous in general. A useful illustration of this is the amplitude damping channel, characterised by $\mathcal{E}:\ket{1}\rightarrow \sqrt{\gamma}\ket{0}+\sqrt{1-\gamma}\ket{1}$. We find that $\chi_{\mathcal{E}} = (1-\gamma/2)\left(\ket{00}+\sqrt{1-\gamma}\ket{11}\right)\left(\bra{00}+\sqrt{1-\gamma}\bra{11}\right) + \gamma/2\ket{10}\bra{10}$. As long as $\gamma<1$, this is a mixture of an orthogonal pure and separable state, meaning the only valid choices for $\sigma_0$ are $\sigma_0=\gamma'/2\ket{10}\bra{10}$, $\gamma'\leq \gamma$. In particular, this means $d\,\sigma^0_{A}$ is rank 1, and thus $\alpha = 1 =\log 2$ as long as $\gamma<1$. When $\gamma=1$ however, $\chi_{\mathcal{E}}$ is separable; implying $\sigma_0=\chi_{\mathcal{E}}$, $d\,\sigma^0_{A}=\openone$, $\alpha=0$. For other common channels though the measure captures a more fine-tuned notion of dimensionality than the Schmidt number; for the depolarising channel, we see that $D_{M}(\mathcal{E}_p)= (1-3p/2)$ when $p<2/3$, and 0 otherwise; whereas $S_{N}(p)= 1$ when $p<2/3$. In this case, the measure captures the noise of the channel. For the qubit erasure channel $\mathcal{R}_{q}(\rho):=(1-q)\rho +q\ket{e}\bra{e}$, where $q$ is probability of erasure, $D_{M}(\mathcal{R}_{q})= 1-q$.

\section{Compression of Quantum Measurements}
\subsection{Motivation}
We will now turn to a topic in which the concept of dimensionality naturally plays an important role. This is the topic of quantum storage or quantum memory, the practical construction of which is an active research topic  \cite{ReviewmemorySimon10,ReviewmemoryHeshami16}. As in classical information theory, it is natural that we may wish to store information, so that it may be accessed and used later on in, e.g., computation. This is the role of a \emph{memory}. The amount of storage available is limited, so it is useful to \emph{compress} the information in some way beforehand - encode it into a new state which requires less bits (classical) or qubits (quantum) to store. This compression can be lossless, where the exact state can be reconstructed from the stored information, or lossy, where the original state is only reconstructed approximately, normally with the advantage of lower memory costs. Note that we use the word lossy in the classical information theory sense to describe imperfections in general, which deviates from the terminology used in quantum information theory, where losses typically refer to particle losses.\\

Recently in Bluhm et al. \cite{BRW2018}, the authors considered an application of quantum memory, in which an unknown state is compressed into a quantum memory, and later decompressed, such that a measurement may be performed on that state. This measurement is taken from a set of available measurements, and it is not known at the time of storage which of the available measurements will be chosen. Here, measurements are described by positive operator-valued measures (POVMs for short). For a given input (or measurement choice) $x$, the corresponding POVM is a collection of positive semi-definite matrices $\{M_{a|x}\}_a$ for which $\sum_a M_{a|x}=\openone$. Given a fixed set of available measurements $\bs{M}=\{M_{a|x}\}$, the authors of Ref.~\cite{BRW2018} thus require that
\begin{equation}
\mathrm{Tr}(\rho M_{a|x}) = \mathrm{Tr}(\mathcal{D}\circ\mathcal{E}(\rho) M_{a|x}),\;\;\forall \rho,\label{bothway}
\end{equation}
where $\mathcal{D},\mathcal{E}$ are two CPTP maps, and $\mathcal{E}$ maps the input space to $\oplus_{i}\mathcal{H}_{n_i}$ and $\mathcal{D}$ maps back in to the original space. Here, the space $\oplus_{i}\mathcal{H}_{n_i}$ represents a quantum memory, where each member of the direct sum is a Hilbert space with a smaller dimension $n_i$ than the original space, and the index $i$ is the classical memory. The authors argue that we should consider classical memory as a free resource, due to its easy availability relative to quantum memory, and were successfully able to quantity the amount of quantum memory required for such a task, given a set of measurements $\bs{M}$. Note that, since the state being given is arbitrarily chosen, it is the measurement set which determines the resources required, rather than the states themselves.\\

In this work we wish to use a slightly adjusted version of the protocol, by allowing for \emph{different measurements} to be made after the compression has been done (a possibility noted by the authors of Ref.~\cite{BRW2018} and discussed in more detail in Ref.~\cite{ioannou2022simulability}). However, these measurements should still reproduce the desired measurement statistics, so that
\begin{equation}\label{reqmeas}
\mathrm{Tr}(\rho M_{a|x}) = \mathrm{Tr}(\mathcal{E}(\rho) N_{a|x}),\;\forall \rho
\end{equation}
where $\mathcal{E}$ is an $n$-PEB channel and $\bs{N}=\{N_{a|x}\}$ are the measurements after the compression. In other words, here one has more freedom in that one can perform a direct readout given by general POVMs $\bs{N}$ of the memory system, instead of having to decompress the memory with another map $\mathcal D$ before a specific readout given by $\bs{M}$, cf. Eq.~(\ref{bothway}).\\

This adjustment is well motivated, due to the following example: when $\bs{M}$ is a single POVM, cf. Fig.~\ref{Fig:singlemeasurement}. Here there is no ambiguity about which measurement will be later performed, and so we can perform the measurement immediately, storing the result in classical memory at no cost. However, the requirement in Equation (\ref{bothway}) of a decompression channel $\mathcal{D}$, imposes a non-zero quantum memory requirement, at odds with our reasoning. Another nice property of the adjustment is that for sets of measurements with no quantum memory requirement, one recovers the definition of joint measurability \cite{ioannou2022simulability}, which is a central concept in quantum measurement theory. Joint measurability asks whether there exists a single POVM whose statistics can reproduce the statistics of a set of given POVMs via classical post-processing (i.e. readouts of a classical memory), see \cite{JMinvitation,JMreview} for reviews on the topic.

\begin{figure}[h]
    \begin{subfigure}[b]{0.4\textwidth}
      \includegraphics[width=\textwidth]{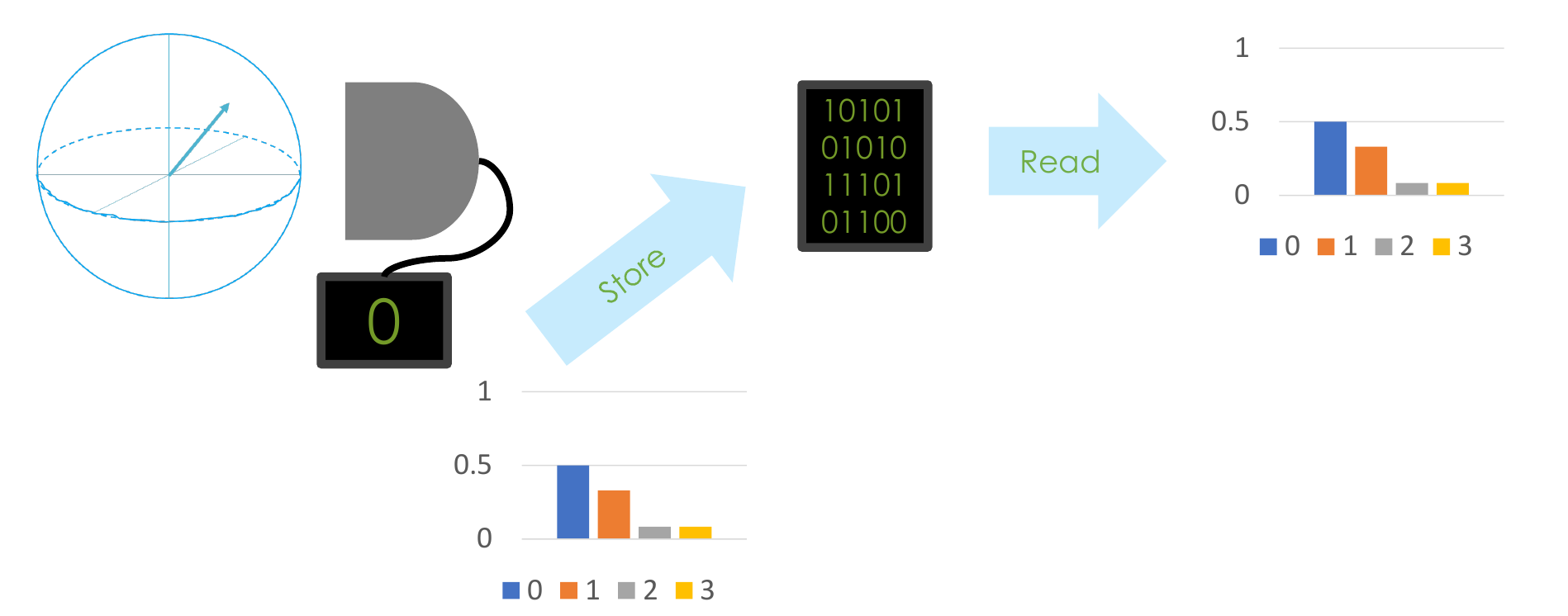}
      \caption{Measuring a single POVM immediately, and storing the result in classical memory, allows us to later give the correct statistics, without using quantum memory}
    \end{subfigure}
    \begin{subfigure}[b]{0.4\textwidth}
      \includegraphics[width=\textwidth]{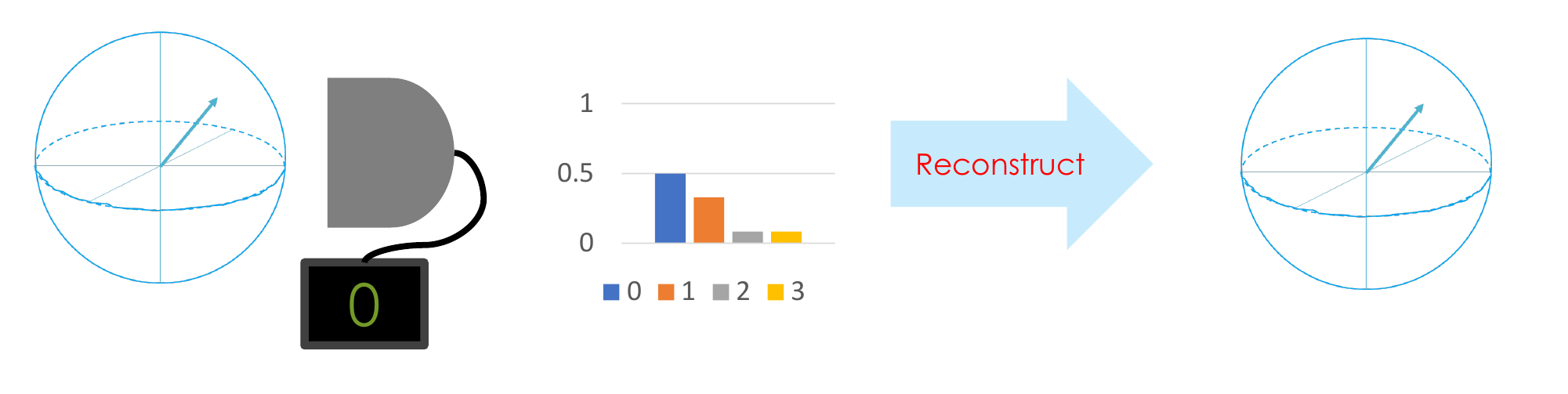}
      \caption{By contrast, if we are required to return a state, to be measured by the original measurement, reconstruction from the classical outcomes is impossible (the map is not CPTP in general).}
    \end{subfigure}
    \caption{A comparison of compression methods to reproduce the statistics of a single POVM.}
    \label{Fig:singlemeasurement}
\end{figure}

As discussed earlier, the dimension of the quantum memory refers to the need to preserve \emph{coherences}, rather than the output dimension of our compression channel $\mathcal{E}$. We note that this is in line with other approaches to quantum memories, which relate memory with the entanglement-preserving properties of the channel \cite{Denis18,Yuan2021,Ku22}. Notably, situations including channels with a Kraus representation using only rank-1 operators are seen as classical in both cases. Let us consider again the case of a single POVM $\bs{M}=\{M_{a}\}$ (note that this covers the case when many POVMs are jointly measurable). We stated that our plan will be to simply measure the state, then store this in classical memory to be read later. This can be written using the explicit channel $\mathcal{E}$ with Kraus operators $\{\ket{a'}\bra{i}\sqrt{M_{a'}}\}_{a',i}$ and measurement defined via: $N_{a}=\ket{a}\bra{a}$.  We see that although the dimension of the output of the overall system could be very high, any output of the channel can be reduced into the direct sum of one-dimensional subspaces, which is equivalent to a classical system. This now gives our connection to the dimensionality of a quantum channel. If Kraus operator $K_i$ is enacted, then $\log{K_i}$ units of memory are required to store the resultant state, before a measurement from $\bs{N}$ is performed. Thus by optimising over the channel and subsequent measurements, we arrive at the following definition of dimensionality: The \emph{compression dimension} $D_{N}(\bs{M})$ of a set of measurements is given by:
\begin{align}
D_{N}(\bs{M}):=&\min_{\mathcal{E}} \bigg\{\max_{i}\log\mathrm{rank}(K_{i})\,\bigg|\; \exists\, \bs{N},\\
\mathcal{E^*}(N_{a|x})
=& \sum_{i} K_i^{\dagger} N_{a|x} K_{i} = M_{a|x}\bigg\},\nonumber
\end{align}
where we have considered the channel in the Heisenberg picture, which rewrites condition (\ref{reqmeas}) equivalently as $\mathrm{Tr}(\rho M_{a|x})=\mathrm{Tr}(\rho\mathcal{E}^{*}(N_{a|x}))$. The quantity $\max_{i}\log \mathrm{rank}(K_{i})$ is aforementioned \emph{Schmidt number} of a channel, if it is minimised over all Kraus representations. We therefore use the notation $D_{N}$ to highlight this similarity. We note that up to the logarithm, this is equal to the concept of $n$-simulability defined in Ref.~\cite{ioannou2022simulability}. We also comment that these compressions can be seen as processings of ``programmable measurement devices" as discussed in \cite{Buscemi2020}, and conversely a processing gives rise to an explicit compression protocol. This means both the compression dimension and memory cost (introduced below) can be seen as incompatibility measures, respecting the partial order given by the resource theory of programmable measurement devices.

However, the above discussion is still not the full story. Imagine a scenario in which $\bs{M}$ is a convex mixture of two measurement sets: a jointly measurable set which is measured with high probability, and an incompatible set which is measured with low probability. To reproduce such a measurement, we need access to a true quantum memory - thus the corresponding compression dimension will be high. However, we only need to use this memory sparsely, implying the resource cost of this measurement should be relatively low. How do we capture this situation? \\

A reasonable solution would be to weight the cost of each stored state by how often it is stored, as we considered for channels. For a given input state $\rho$, the probability of using a particular compression is given by $\mathrm{Tr}(K_{i}^{\dagger}K_{i}\rho)$. Thus if the average input state over a possible set is known, then the memory cost becomes
$\min_{K_i} \sum_i \mathrm{Tr}(K_{i}^{\dagger}K_{i}\rho) \log \mrm{rank}(K_i)$
such that the above relations still hold. However, in general we consider the case where $\rho$ is completely unknown to us; and thus, define the memory cost as:
\begin{equation}\label{costKrausOp}
D_M(\bs{M})=\min_{K_i} \max_{\rho} \sum_i \mathrm{Tr}(K_{i}^{\dagger}K_{i}\rho) \log \mrm{rank}(K_i),
\end{equation}
i.e. we choose whatever compression protocol we like, and measure its performance upon the most demanding input state. As this is a compression protocol, Eq. \ref{costKrausOp} implicitly requires a measurement set $\bs{N}$, which along with the channel reproduce our original measurements. We use the notation $D_{M}$, or ``dimension measure", to refer to this quantity. We may use here minimum and maximum, rather than supremum/infimum; a proof is given in Appendix \ref{DimMeasMeas}. We note that our measure differs from a related dimensionality measure for POVMs given in Ref.~\cite{ioannou2022simulability}, in that Ref.~\cite{ioannou2022simulability} considers the maximal Kraus rank needed instead of the average. As an example, sets of measurements that are only slightly incompatible result in a dimension-measure that is greater than or equal to (logarithm of) two when using the approach of Ref.~\cite{ioannou2022simulability}. In our approach the measure takes values that are smaller than or equal to the ones given by the maximal Kraus rank, hence, better encapsulating the fact that the measurements are only slightly incompatible. We further note that the way we do the averaging consists of a non-trivial choice of the probability distribution used in the definition. One could consider other choices such as $\text{tr}[K_{i}^{\dagger}K_{i}]/d$, but we found that such choices do not result in a sub-additive measure under tensor products. To illustrate this, consider the incompatible qubit measurements $M_0=\{\ket{0}\bra{0},\ket{1}\bra{1}\}, M_1=\{\ket{+}\bra{+},\ket{-}\bra{-}\}$. Both our choice of probability distribution and the choice mentioned above give a value of $\log 2 =1$. However, if we append an additional outcome to each measurement, $\tilde{M}_0=\{\ket{0}\bra{0},\ket{1}\bra{1},\ket{2}\bra{2}\}, M_1=\{\ket{+}\bra{+},\ket{-}\bra{-},\ket{2}\bra{2}\}$, this alternate choice drops to $2/3$, whilst our choice remains at 1. Since we can reproduce $M_0, M_1$ by restricting to the qubit subspace, we see that our choice, requiring a maximisation over input states, is the correct one.

It is useful to consider \emph{how} such a compression can be implemented; in particular, how we know the amount of memory to assign. This can be done in the following way: the channel is explicitly implemented via the isometric extension defined by $U= \sum_{i} K_{i}\otimes \ket{i}_{E}$. By measuring the ancillary system $E$, we know which Kraus operator has been applied; and thus how much memory to assign. This protocol is visualised in Figure \ref{compression_fig}. Tracing over the ancillary system reproduces the overall action of the channel. It is worth noting here that, since we know the outcome of $i$, we can in principle perform different measurements $N_{a|x}^{i}$ for each outcome. This does not add any additional power to the formulas considered above; since we may always define $K_{i}'=K_{i}\otimes \ket{i}_{E}$, $N_{a|x}=\sum_{i} N_{a|x}^{i}\otimes \ket{i}\bra{i}$.\\

\begin{figure}
    \centering
    \includegraphics[width=\columnwidth]{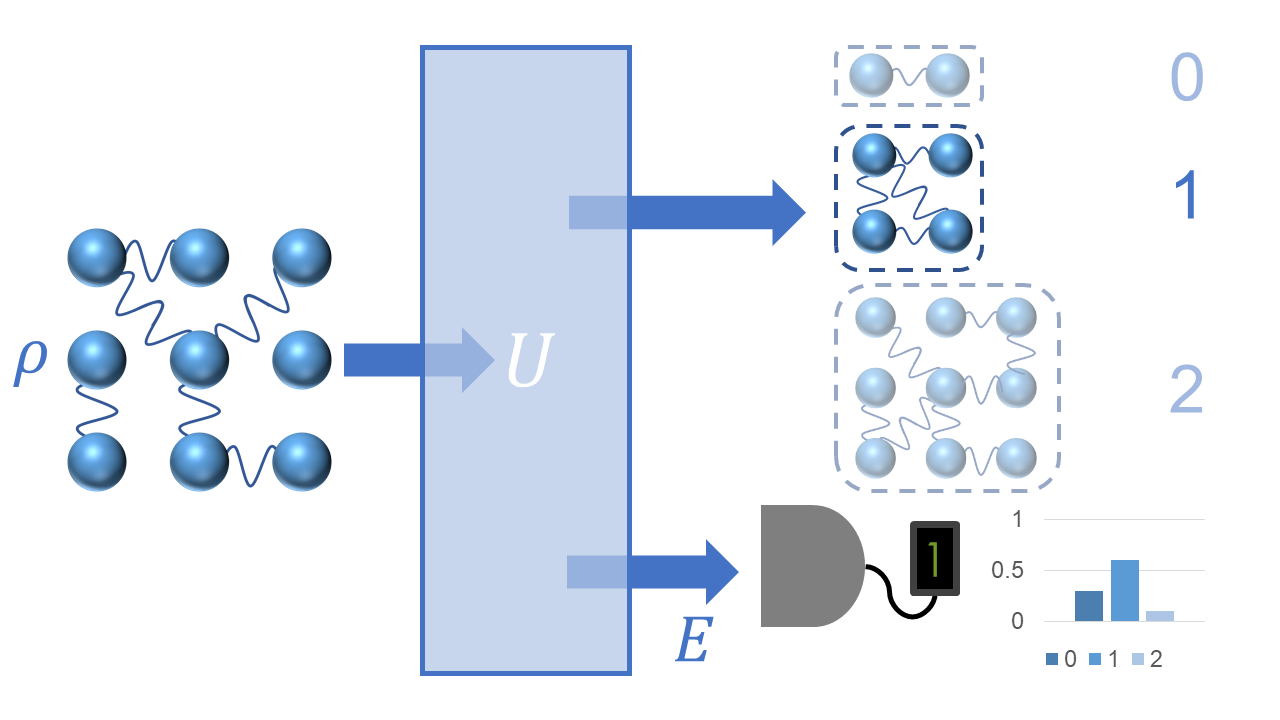}
    \caption{The probabilistic compression protocol: A high dimensional state $\rho$ is subjected to an isometry $U$. The environment is then measured, telling the user how much memory to assign to the conditional state $K_i\rho K_i^{\dagger}$. This compressed state can be later measured to reproduce the statistics of the desired measurement on $\rho$. Although we represent here high dimensional states with multiple qubits, this protocol is applicable to general degrees of freedom.}
    \label{compression_fig}
\end{figure}

The above quantity satisfies the following desirable properties:
\begin{itemize}
\item It is convex under combination of sets of measurements, i.e. under classical pre-processing.
\item It is subadditive over ``tensor products''.
\item It is non-increasing under classical post-processing and quantum pre-processing.
\item Jointly measurable measurements have dimension measure 0.
\end{itemize}
The use of the quotation marks around tensor product is because we have not yet rigorously defined what this is. The definition is done in such a way to emphasize the measurements take place on different systems; and that a measurement choice is available for each individual system. Thus $\{M_{a|x}\}\otimes \{N_{b|y}\} = \{M_{a,b|x,y}:=M_{a|x}\otimes N_{b|y}\}$. It is over this construction that the dimension measure is subadditive.\\

Although this definition of the dimension measure is operationally useful to us,
trying to optimise over channels, Kraus operators and $\bs{N}$ is a challenging task. We can reformulate the problem in an entirely equivalent way, which simplifies this slightly. The dimension measure can be equivalently written as:

\begin{equation}
D_{M}(M_{a|x})=\min_{A^i} \max_{\rho} \sum_{i} \mathrm{Tr}(A^i\rho)\log \mrm{rank}(A^i)
\end{equation}
where $0\leq A^{i}_{a|x}$ such that 
$\sum_{i} A^{i}_{a|x} = M_{a|x}$, $\sum_{a} A^{i}_{a|x} = A^{i}$.\\

This definition is more satisfying in that it is a property solely of the measurement operators themselves. However, the two formulisms are associated via $A^{i} =K_{i}^\dagger K_{i}$, $A^{i}_{a|x}=K_{i}^{\dagger}N_{a|x}K_{i}$. 
This rewriting allows us to state another property: if our measurement operators are decomposable into block diagonal form $M_{a|x}=\oplus M^{i}_{a|x}$, then the compression dimension and measure are upper bounded by $d_i$, the size of the largest block.\\

    Importantly, the marginal operators $A_i$ are \emph{not} required to be the identity, or even the identity on a smaller subspace. We shall therefore refer to $\bs{A}^i:=\{A^{i}_{a|x}\}$ as (a set of) ``pseudo-measurements". This allowance is vital to give the right quantity. 
    
Another useful writing of this formula is to group all the $\bs{A}^{i}$ (in the optimal decomposition) of the same rank together. We will typically write this decomposition as
\begin{equation}\label{grouping}
M_{a|x}=\sum_{k} B^{k}_{a|x},
\end{equation}
with the dimension measure given by $D_M(\bs{M})=\max_{\rho}\sum_{k}\mathrm{Tr}(B^{k}\rho)\log k$.

\subsection{The Use of Symmetry}

When the measurement set satisfies certain symmetries, we can often restrict or even remove the optimisation over input states $\rho$. If there exist unitaries $\{U_j\}_{1}^{N}$ such that $U_{j}^{\dagger}M_{a|x}U_{j} = M_{\pi_{j,x}(a)|\pi_j(x)}$, where $\pi_j$ is a permutation of the inputs, and each $\pi_{j,x}$ is a permutation of the outputs, then the optimal state must be of the form $\rho=\frac{1}{N}\sum_{j}U_j \rho'U_{j}^{\dagger}$.\\

A particularly useful case is when the unitaries $U_j$ are the set $\{X^{i}Z^{j}\}_{i,j=0}^{d-1}$. Then we have that $\rho=\openone/d$, see Appendix \ref{PropMeasDim} for details. This means the dimension measure is given by the \emph{Schmidt measure} of the (optimally chosen) channel. A naturally interesting set of measurements that is invariant under these operators is the set of $d+1$ (in prime dimensions) mutually unbiased bases given by the eigenvectors of $\{XZ^{j}\}_{j=0}^{d-1}\cup\{Z\}$, mixed with white noise. In fact, any \emph{subset} of these $d+1$ mutually unbiased bases is invariant under these operators, since our unitaries only permute each individual basis. This means for measurements of the form:
\begin{equation}
M^{p}_{a|x} = p\ket{e^x_a}\bra{e^x_a}+(1-p)\openone/d \label{whitenoise}
\end{equation}
where $\ket{e^x_a}$ is the $a$\textsuperscript{th} eigenvector of unitary $U_x$, then taking $\rho=\openone/d$ is sufficient.\\

One case of particular interest is a pair of MUBs: taking the eigenvectors of $\{X,Z\}$ only. For such a scenario we can go a step further and define the unitaries $\{U_j\}$ which, when averaged over with the corresponding permutations, will generate measurements of the form of Eq. (\ref{whitenoise}) from an arbitrary pair of pseudo-measurements (this construction is described in Appendix \ref{MUBSym}). Calculating the dimension measure then relies on analysing this one-parameter family; in particular, one needs to construct the least noisy MUB pair possible from rank $k$ pseudo-measurements. Our heuristic construction to do this is as follows. Construct the operator $S_k=\sum_{a\in\mathcal{O}_1} \ket{a_X}\bra{a_X}+\sum_{a\in\mathcal{O}_2}\ket{a_Z}\bra{a_Z}$, where $\mathcal{O}_1,\mathcal{O}_2\subset \{1\ldots d\}$ are of order $k$. Then set $A_{a|i}:=\begin{cases}
\Pi_k\ket{a_i}\bra{a_i}\Pi_k, &\text{if } a\in \mathcal{O}_i,\\
0 &\text{ else,}
\end{cases}$,
where $\Pi_k$ is the projector onto the $k$ largest eigenvalues of $S_k$.\\
Averaging over symmetries then gives an assemblage $\bf{M}^p$, which for $k=1$ can be verified to be the largest jointly measurable $p$ value, cf.~\cite{CHT12,Haa15,Uola16}. For larger $k$, optimisation over the choice of $\mathcal{O}_i$ is required. Using the construction for general $k$, it seems the dimension measure for general $p$ is always minimised by mixing $k=1$ and $k=d$ together; it is unclear if there exist improved compression dimension $k$ constructions, which could lower the dimension measure \footnote{Preliminary results from HC Nguyen and J. Steinberg suggest there exist cases where this construction is not optimal, but the SDP-upper bound is nevertheless the true value  (private correspondence).}. Thus we leave as an open problem the following conjecture :
\emph{For the pair of $d$-dimensional MUBs given by the eigenvectors of $\{X,Z\}$, the above construction gives the largest $p$ value (visibility) possible via decomposition into compression dimension $k\leq d$ pseudo-measurements. Moreover, we conjecture the dimension measure is given by the SDP-obtained upper bound.}\\
We shall see how to obtain the upper bound in the following section. A comparison of the upper bound on the dimension measure obtained via SDP vs. the value given by our explicit constructions is presented in Figure \ref{MUBGraph}.

\begin{figure}
    \centering
    \includegraphics[width=\columnwidth]{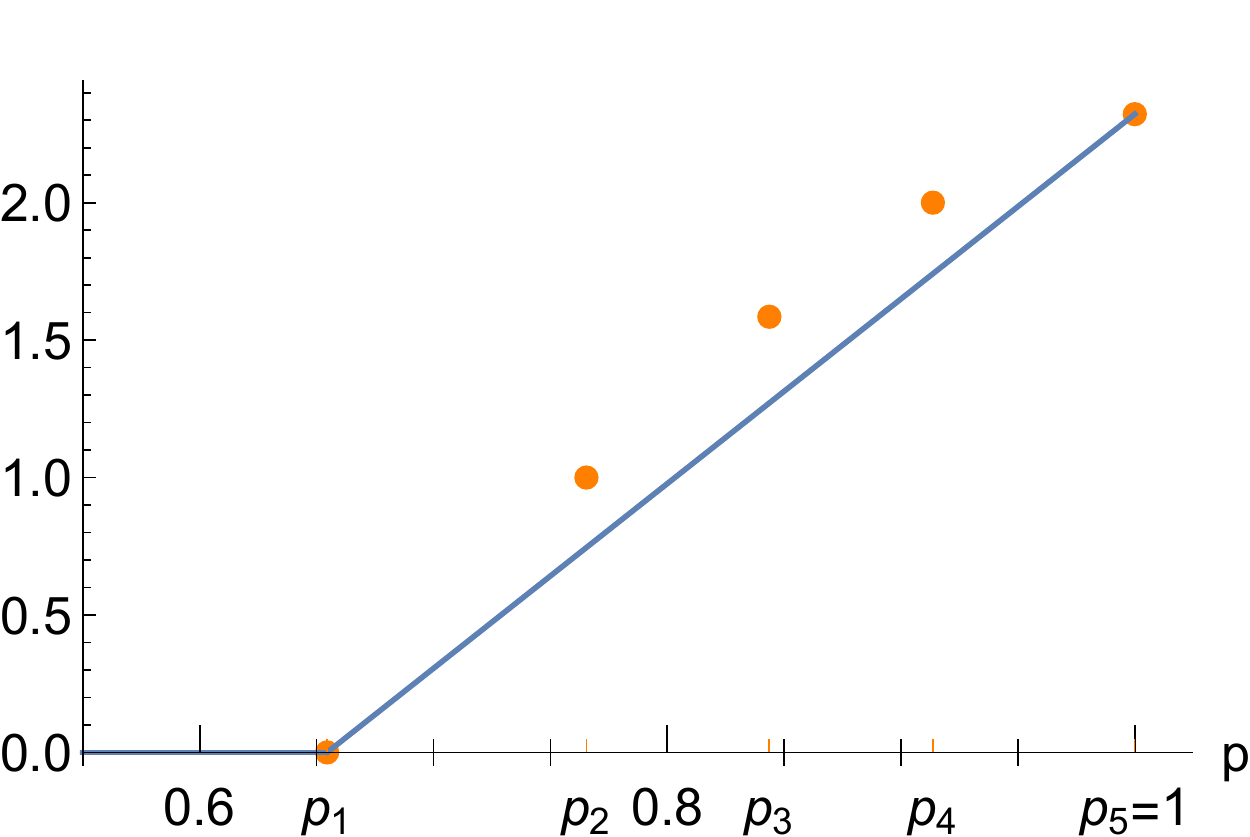}
    \caption{Looking at the projective measurements on the $X$ and $Z$ operators for $d=5$, we see the SDP-based upper bound beats our explicit constructions of the measurements from those of lower compression dimension.}
    \label{MUBGraph}
\end{figure}

\subsection{Compression of Qubit Measurements}
In Equation (\ref{grouping}) we grouped together pseudo-measurements by their rank.
Of particular interest is the set of possible $\bs{B}^0$ - linear combinations of $\bs{A}^{i}$ with rank 1. This forms the set of \textit{jointly measurable pseudo-measurements}, $\mathcal{B}_{\mathrm{JM}}:=\{\mathbf{B}\mid B_{a|x}=\sum_{\lambda} D(a|x,\lambda)G_{\lambda},\;\sum_{\lambda}G_{\lambda}\leq \mathbb{I}\}$, and has the nice property that any member of the set can be represented by a set of positive operators $\{G_{\lambda}\}$, with $\lambda$ running over all deterministic assignments of an outcome to each input choice. $D(a|x,\lambda)$ above is the deterministic probability distribution corresponding to this choice. Consider Equation~(\ref{grouping}) in the case where our original measurements are on qubits, i.e. $M_{a|x}\in \mathcal{B}(\mathbb{C}_2)$. Then the decomposition is into $\bs{B}^0 + \bs{B}^{1}$ only, and the dimension measure is given by:
\begin{equation}
\min_{\bs{B}^0} \max_\rho \mathrm{Tr}(B^1\rho),
\end{equation}
such that $\bs{B}^0=\{B^0_{a|x}\}\in \mathcal{B}_{\mathrm{JM}}$ is jointly measureable. Furthermore, we know in this case that $B^{1}=\openone_2-B^{0}$, meaning we can rewrite the measure as 
$\min_{\bs{B}^0} \max_\rho 1-\mathrm{Tr}(B_0\rho)\equiv\min_{\bs{B}^0} 1- \|B_0\|_{\mathrm{op}}$. This quantity can be succinctly written as the result of the following SDP:\\
\begin{align}
&\text{minimize}\; \alpha,\label{DimMeasMeasSDP}\\
&\text{subject to:}\nonumber\\
& G_{\lambda}\geq 0\, \forall\,\lambda,\nonumber\\
&\sum_{\lambda(x)=a} G_{\lambda} \leq M_{a|x},\nonumber\\
&(\alpha-1) \openone_2 + \sum_{\lambda} G_{\lambda} \geq 0\nonumber
\end{align}
where one notes that $\sum_{\lambda} G_{\lambda}=B^0$. This means that, for an arbitrary set of qubit measurements, one can calculate the dimension measure \emph{exactly}. Figure \ref{MUBGraph} shows the dimension measure for two and three noisy Pauli measurements.\\
\begin{figure}
    \centering
    \includegraphics[width=\columnwidth]{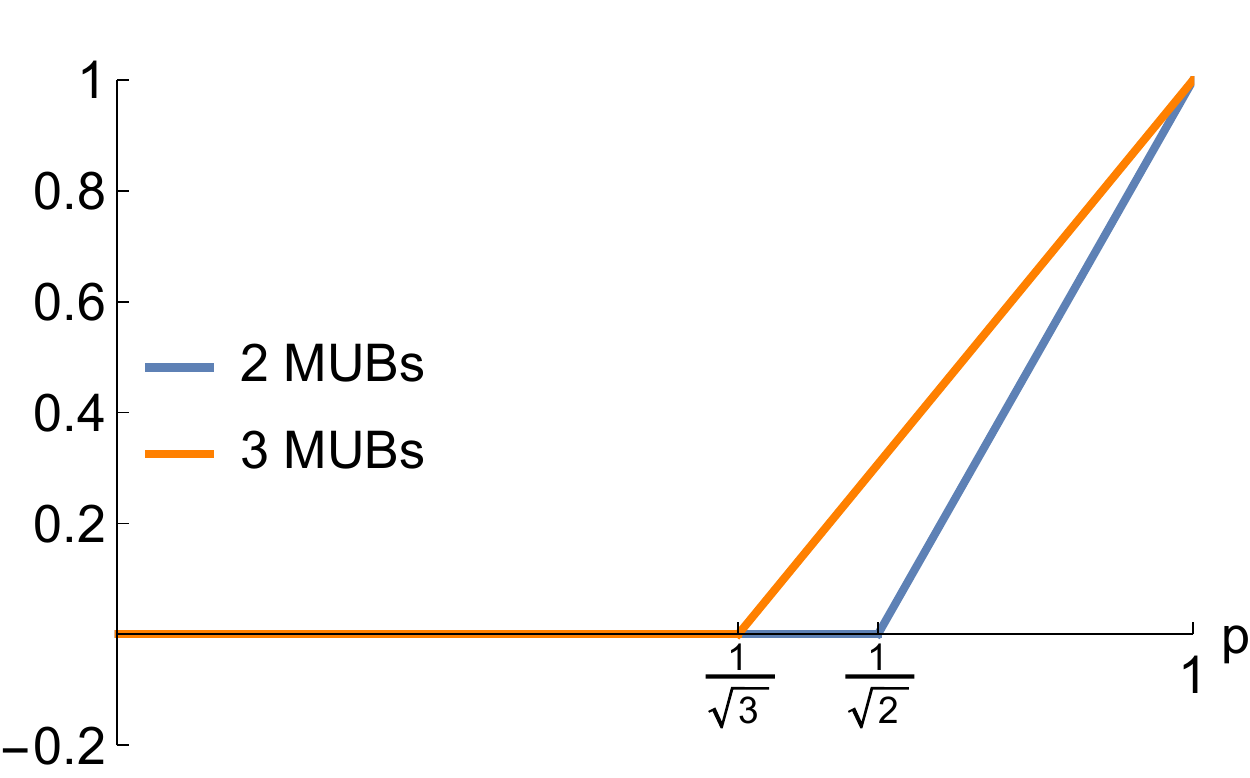}
    \caption{The dimension measure for two and three Pauli measurements respectively, mixed with white noise. The requirement to reproduce an additional measurement increases the memory cost.}
    \label{MUBGraph}
\end{figure}

It is also worth comparing this SDP to that which calculates the \emph{incompatibility weight} of a set of measurements, $W(\bs{M})$ \cite{Pusey15}. This is given as:
\begin{align}
&\text{minimize}\; \gamma,\\
&\text{subject to:}\nonumber\\
& G_{\lambda}\geq 0\, \forall\,\lambda,\nonumber\\
&\sum_{\lambda(x)=a} G_{\lambda} \leq M_{a|x},\nonumber\\
&(\gamma-1) \openone_2 + \sum_{\lambda} G_{\lambda} = 0.\nonumber
\end{align}
The two share a common structure; with the only difference being that the jointly measurable component for the incompatibility weight must be a true (non-pseudo) measurement. Since this is a more restrictive condition, we see that this provides an upper bound on the dimension measure; and there exist POVM sets for which the dimension measure is strictly smaller. An example of this is given in the accompanying code repository \footnote{\url{https://github.com/ThomasPWCope/Quantifying-the-high-dimensionality-of-quantum-devices}}. In the case of non-qubit measurements, one can always use the SDP in Equation (\ref{DimMeasMeasSDP}) to obtain an upper bound for the dimension measure: for higher dimension, $\openone-B^0$ is the sum of the other $B^{k}$. By assigning them all to have maximum dimension (since we do not know the true values) we obtain the upper bound $\alpha\log d$, where $\alpha$ is calculated by the SDP above, except with $\openone_d$ replacing $\openone_2$.

\section{Schmidt measure for Assemblages}
We now turn our attention to \emph{assemblages}; a set of subnormalised quantum states $\sigma_{a|x}$ is called an \emph{assemblage} if they satisfy:
\begin{equation}
\sum_{a} \sigma_{a|x} = \rho_{\boldsymbol{\sigma}}\;\forall x,\;\;\sigma_{a|x}\geq 0, \mathrm{Tr}(\rho_{\boldsymbol{\sigma}})=1.
\end{equation}
Typically they arise in \emph{steering scenarios}, as the correlations between a trusted and untrusted quantum system. Measurements on the untrusted system give rise to assemblages via the relation $\sigma_{a|x}=\mathrm{Tr}[(M_{a|x}\otimes \openone)\rho_{AB}]$. We will use as shorthand $\boldsymbol{\sigma}$ for the set $\{\sigma_{a|x}\}$.

A state assemblage is called \emph{unsteerable} if it can be explained by a local ensemble $\{p(\lambda),\sigma_\lambda\}$ whose priors $p(\lambda)$ are updated according to the classical information about the measurement choice $x$ and the measurement outcome $a$ available to the parties. Formally, this amounts to $\sigma_{a|x}=p(a|x)\sum_\lambda p(\lambda|a,x)\sigma_\lambda$. We note that by using the Bayes rule, one can rewrite this in the form $\sigma_{a|x}=\sum_\lambda p(a|x,\lambda)\sigma_\lambda$, cf. Refs.~\cite{wiseman2007steering,cavalcanti2016quantum,uola2020quantum}. An assemblage, for which such a \emph{local hidden state model} does not exist is called steerable.

Although the nature of the untrusted system means that, in general, there is no way of knowing the underlying state $\rho_{AB}$ fully, we can nevertheless use our knowledge of the assemblage to say something about it. In particular, if an assemblage $\bs{\sigma}$ is extremal, then the Schmidt measure of the state used to create that assemblage must be $\mathrm{rank}(\rho_{\bs{\sigma}})$. This is because a decomposition of our underlying state $\rho_{AB}$ into pure states $\rho_{AB}=\sum_i p_i|\phi_i\rangle\langle\phi_i|$ gives rise to a convex decomposition of our assemblage $\bs{\sigma}=\sum_i p_i\bs{\tau^i}$, where $\tau^i_{a|x} = \mathrm{Tr}(M_{a|x}\otimes \openone \ket{\phi_i}\bra{\phi_i})$.
However, due to extremality, each $\bs{\tau}^{i}$ must coincide with $\bs{\sigma}$. Therefore $\rho_{\bs{\tau}}=\rho_{\bs{\sigma}}$.  As $\mathrm{Tr}_{A}(\ket{\phi_i}\bra{\phi_i})=\rho_{\bs{\tau}}$, we see that each pure state in the decomposition has Schmidt rank equal to $\mathrm{rank}(\rho_{\bs{\sigma}})$, regardless of the choice of decomposition. Thus we can conclude that this is the Schmidt measure value of $\rho_{AB}$. In the special case where our assemblage is extremal, we can read off easily what the Schmidt measure of the state must be.\\

If our assemblage is not extremal; then a decomposition $\bs{\sigma}=\sum_i p_i\bs{\tau}^i$ into extremal assemblages necessarily gives us a realisation ($M_{a|x},\rho_{AB}$), where the Schmidt measure of $\rho_{AB}$ is given by $\sum_{i}p_i\log \mathrm{rank}(\rho_{\bs{\tau}^i})$. This idea is illustrated in Figure \ref{Decomp_of_Assem}. Conversely, a realisation $(M_{a|x},\rho_{AB})$ gives a decomposition, via a corresponding decomposition of $\rho_{AB}$, of $\bs{\sigma}$ into extremals \footnote{Although the decomposition of $\rho$ into pure states does not necessarily lead to extremal assemblages, any further decomposition will preserve the marginal, and thus our quantity of interest} where $\sum_{i}p_i\log \mathrm{rank}(\rho_{\bs{\tau}^i})$ is the Schmidt measure. Given this equivalence, one can define the Schmidt measure of an assemblage as the minimum Schmidt measure over all underlying states, such that the assemblage is realisable. This gives
\begin{equation}
S_M(\bs{\sigma}) = \min_{p_i,\bs{\tau}^i} \sum_i p_i \log \mathrm{rank} \rho_{\bs{\tau}^i},
\end{equation}
such that $\sum p_i \bs{\tau}^i = \bs{\sigma}$.
If one instead considers $\min_{p_i,\bs{\tau}^i} \max_i \mathrm{rank} \rho_{\bs{\tau}^i}$, then one obtains the notion on $n$-preparability introduced in \cite{Detal2021,UKSYG2019}, which describes the minimal Schmidt number of $\rho_{AB}$ required. Hence, our notion of Schmidt measure for state assemblages relates to $n$-preparability in a similar manner as our dimension measure of POVMs relates to the concept of $n$-simulability \cite{ioannou2022simulability}.\\
\begin{figure}
    \centering
    \includegraphics[width=\columnwidth]{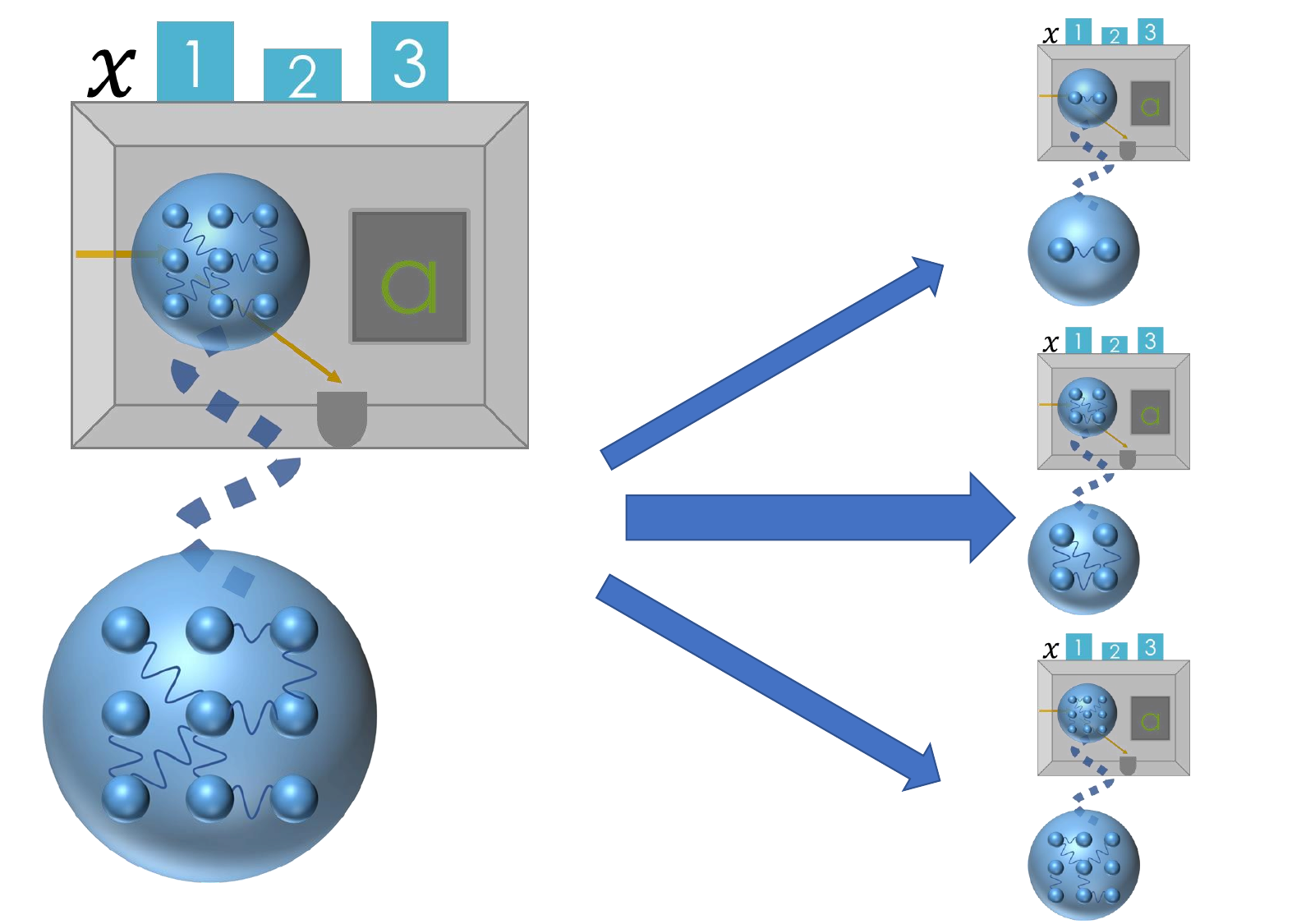}
    \caption{Decomposition of assemblages. A convex decomposition of an assemblage implies a convex decomposition of the underlying state. As a consequence of this, we can determine the minimum Schmidt measure to create the assemblage.}
    \label{Decomp_of_Assem}
\end{figure}
Note that, as with channels and measurements; in the specific case where our steering assemblage consists of qubit substates, then this quantity is calculable via the following SDP:
\begin{align}
&\text{minimise}\; 1-\mathrm{Tr}\left(\sum_{\lambda} \tau_{\lambda}\right),\\
&\text{subject to:}\nonumber\\
& \tau_{\lambda}\geq 0\; \forall\,\lambda,\nonumber\\
&\sum_{\lambda(x)=a} \tau_{\lambda} \leq \sigma_{a|x}.\nonumber
\end{align}
We note that in this case the optimization problem corresponds exactly to that given in \cite{Steeringweight} for the steering weight.\\
As with the dimension measure, we obtain an upper bound for the Schmidt measure for higher dimensional assemblages, by using the general cost function $\log d\left(1-\mathrm{Tr}\left[\sum_{\lambda} \tau_{\lambda}\right]\right)$. Similar to the situation with incompatible measurements, the prefactor of $\log d$ appears by our simplification of the problem, which maximises the unsteerable part (possible using a separable state) and assigns the remainder the maximum possible dimension. It is easy to construct cases, however, where this upper bound does not come close to the true value. For example, consider the assemblage $\sigma_{a|0}=\ket{a}\bra{a}/d$, and $\sigma_{a|1}=(\ket{+_a}\bra{+_a}+\ket{-_{a\oplus 1}}\bra{-_{a\oplus 1}})/2d$, with $\ket{\pm_a}=\ket{a}\pm\ket{a\oplus 1}$. Here we use $\oplus$ as shorthand for addition $\mod d$. From \cite{CO2022}, we know extremal unsteerable assemblages to be of the form $\tau_{a|x}=\delta_{aa_x}\ket{\phi}\bra{\phi}$, where $a_x$ is an explicit choice of outcome for each individual $x$. As there exists no pair $\{\sigma_{a|0},\sigma_{a'|1}\}$ which share a non-trivial subspace, there cannot exist a non-zero unsteerable component and the SDP will therefore return the maximal bound of $\log d$. We can, however, decompose this assemblage into a uniform mixture of extremal qubit assemblages $\bs{\tau}^a$, with non-zero components: $\tau^a_{a|0}=\ket{a}\bra{a}/2d$, $\tau^a_{a\oplus 1|0}=\ket{a\oplus 1}\bra{a\oplus 1}/2d$, $\tau^a_{a|1}=\ket{+_a}\bra{+_a}/2d$, $\tau^a_{a\oplus -1|1}=\ket{-_a}\bra{-_a}/2d$. As we have an explicit decomposition into qubit assemblages, and we have ruled out the possibility of lower rank (i.e. unsteerable) components, we may conclude the true Schmidt measure of this assemblages is $\log 2 \equiv 1$. By increasing $d$, we can make the gap between our upper bound ($\log d$) and true value (1) arbitrarily large. 

Sets of measurements and assemblages are closely connected \cite{uola15}, and in some cases the compression dimension of measurements reduces to an assemblage Schmidt measure. If, due to the symmetry of our measurements, we can take $\rho$ in the definition of the dimension measure to be $\openone/d$, then our dimension measure is:
\begin{align*}
D_M(\bs{M})=&\min_{\bs{A}^i} \sum_{i} \mathrm{Tr}(A^{i}/d)\log\rank(A^{i}),& \sum_{i} A^{i}&=\openone\\
\equiv S_M(\bs{\sigma}) =& \min_{\bs{\tau}^{i}} \sum_{i} p_i\log\rank(\rho_{\bs{\tau}^{i}}),& \sum_{i} \rho_{\bs{\tau}^{i}}&=\rho_{\bs{\sigma}}
\end{align*}
where we associate $\sigma_{a|x} = M_{a|x}/d$, and $\tau^i_{a|x}=A^{i}_{a|x}/\mrm{Tr}(A^{i}_{a|x}),\; p_{i}=\mrm{Tr}(A^{i}_{a|x})/d$. An immediate consequence of this is that a set of two or more MUBs in dimension $d$ has compression dimension $d$. We could also consider, for less symmetrical measurements, the dimension measure as a quantity of assemblages; however, we have not found an operational meaning for this.\\

One can also create measurement sets from assemblages, via the use of ``pretty good measurements''; we recall that the pretty good measurements \cite{Hausladen94} are given by $M_{a|x}:=\rho_{\bs{\sigma}}^{-1/2}\sigma_{a|x}\rho_{\bs{\sigma}}^{-1/2}$ and a pseudoinverse is used when necessary. The connection between steering and measurement incompatibility is straightforward: an assemblage $\bs{\sigma}$ is unsteerable if and only if its pretty good measurements are jointly measurable \cite{uola15}. This hints that the dimensionality of an assemblage and its corresponding pretty good measurement should be closely connected. Suppose we have the optimal decomposition $\bs{\sigma}=\sum_{i}p_i\bs{\tau}^i$ - let us redefine $\tilde{\bs{\tau}}^{i}=p_i\bs{\tau}^i$. Then the Schmidt measure of $\bs{\sigma}$ is given by 
$\sum_i \mrm{Tr}(\rho_{\tilde{\bs{\tau}}^i})\log\rank\rho_{\tilde{\bs{\tau}}^i}$. We see that this decomposition gives rise to a decomposition of the pretty good measurement, via the connection $A^{i}_{a|x}:=\rho_{\bs{\sigma}}^{-1/2}\tilde{\bs{\tau}}^i_{a|x}\rho_{\bs{\sigma}}^{-1/2}$. We can then rewrite the Schmidt measure of the assemblage as $\sum_i \mrm{Tr}(\rho_{\tilde{\bs{\tau}}^i})\log\rank\rho_{\tilde{\bs{\tau}}^i} = 
\sum_i \mrm{Tr}(\rho_{\bs{\sigma}} A^i)\log\rank(\rho_{\bs{\sigma}}A^i)= \sum_i \mrm{Tr}(\rho_{\bs{\sigma}} A^i)\log\rank(A^i)$ 
where the last equality comes from the construction of $A^i$, which tells us $A^{i}$ must lie in the support of $\rho_{\bs{\sigma}}$. We can conclude from this that 
$S_M(\bs{\sigma})=\min_{A_i}\sum_i \mrm{Tr}(\rho_{\bs{\sigma}} A^i)\log\rank(A^i) \leq D_M(\bs{M})$.\\

\section{The Asymptotic Schmidt measure}
All of the quantities discussed so far can be described as \emph{one-shot} quantities. Also of interest in (classical and quantum) information theory is the \emph{asymptotic} cost; the cost-per-instance of an operation in the limit of performing it an infinite number of times. This allows ``economies of scale'' to come into play; as one performs a higher number of the task, one can do so more and more efficiently, usually with a small error which vanishes in the limit.\\

In this section we determine how the Schmidt measure behaves in the asymptotic limit of an infinite number of copies of the state. We show that it converges to the \emph{entanglement cost} of the state. Operationally, this tells us that the probabilistic one-shot entanglement cost converges asymptotically to the entanglement cost (as expected), and that the dimensionality of a state converges to its entanglement. To show these results we expand the definition of the Schmidt measure to a \emph{smoothed} version: $S_M^{\epsilon}(\rho):=\inf_{\|\rho-\rho'\|_1\leq \epsilon} S_M(\rho')$, where we allow for some perturbative noise.\\

Let us recap the definitions of the entanglement cost and the entanglement of formation. The entanglement cost of a mixed state $\rho$ is given by \cite{HHT2001}:
\begin{align}
E_C(\rho):= &\inf\{E\mid \forall \epsilon,\;\; \delta >0\;\; \exists m,\,n,\, \Lambda,\\
\;\; &| E-m/n|\leq \delta,\;\; \|\Lambda(\ket{\Phi^+}\bra{\Phi^+}^{\otimes m}),\rho^{\otimes n}\|_{\mathrm{tr}} \leq \epsilon\},\nonumber
\end{align}
where $\Lambda$ is some LOCC (local operations and classical communication) protocol. This quantifies the optimal rate in which maximally entangled pairs may be converted into copies of our chosen state, such that in the limit of enough copies created the error becomes arbitrarily small.\\

The entanglement of formation of a mixed state $\rho$ is given by:
\begin{equation}
E_F(\rho):= \inf_{p_i,\Phi_i} \sum_{i} p_i E_S(\Phi_i)
\end{equation}
such that $\sum_i p_i \ket{\Phi_i}\bra{\Phi_i} = \rho$, where $E_S(\ket{\Phi})=S(\mathrm{Tr}_{B}[\ket{\Phi}\bra{\Phi}])$ is the entanglement entropy. The entanglement of formation gives an upper bound to the entanglement cost as it gives a specific LOCC protcol to create copies of $\rho$. Each pure state is created by LOCCs from maximally entangled states at the optimal rate of $E_S(\Phi)$ \cite{BBPS1996}. These are then probabilistically mixed together to create the copies of $\rho$. In fact, there exists a stronger result: $\lim_{n\rightarrow \infty} E_F(\rho^{\otimes n})/n = E_C(\rho)$ \cite{HHT2001}. The left hand side is called the \emph{regularisation} of the entanglement of formation. This result tells us that asymptotically, this pure-state mixing is the optimal protocol for creating copies of our state. It also introduces us to the concept of regularisation. It is exactly the regularisation of the Schmidt measure that we will look at in this section.\\

In order to prove our results, we will need some concepts from classical information theory, concerning \emph{typicality} of a sequence. This tries to capture how, when drawing from an independent and identically distributed probability distribution many times, a ``typical'' sequence of outcomes will look. We will use the notion of the \emph{strongly typical set}; suppose we have a probability distribution over a finite alphabet $\mathcal{X}$. Then for sequences of length $n$, the $(\delta)$ strong typical set is defined as \cite{Stanford}:
$T_{\delta}(P):=\{\mathbf{x}\bs{\mid} |N[x|\mathbf{x}]/n-P(x)|\leq \delta P(x)\}$,
where $N[x|\mathbf{x}]$ is the number of times the character $x$ appears in the sequence $\mathbf{x}$. This can be extended to pairs of finite random variables $(X,Y)$, giving the \emph{jointly} typical set of sequences of length $n$:
$T_{\delta}(X,Y):=\{(\mathbf{x},\mathbf{y})\mid |N[(x,y)|(\mathbf{x},\mathbf{y})]/n-P(x,y)|\leq \delta P(x,y)$,
where $N[(x,y)|(\mathbf{x},\mathbf{y})]$ is the number of times the characters $(x,y)$ appear in the same position of the sequences $\mathbf{x},\mathbf{y}$ respectively. We also obtain the  \emph{conditionally} typical set of sequences of length $n$:
$T_{\delta}(Y|\mathbf{x}):=\{\mathbf{y}\mid (\mathbf{x},\mathbf{y})\in T_{\delta}(X,Y)\}$.
Note how this set is dependent on the choice of sequence $\mathbf{x}$. \\

If a sequence $\mathbf{x}$ is not typical, then necessarily $T_{\delta}(Y|\mathbf{x})$ is empty. However, if $\mathbf{x}$ is typical (with respect to X, some $\delta'<\delta$) then for all $n>N$ for some $N$, we have bounds of the size of the set:\\
$(1-\delta)2^{nH(Y|X)\left(1-\delta\right)}\leq |T_{\delta}(Y|\mathbf{x})|\leq 2^{nH(Y|X)\left(1+\delta\right)}$.\\
We also have information about how likely it is for a sequence to belong to the typical set:  $\forall \epsilon$, there exists a $N'$, such that for $\mathbf{x}$ typical when $n>N'$, $P(\mathbf{Y}\in T_{\delta}(Y|\mathbf{x}))>1-\epsilon$. This can be taken a step further, as the result holds that $\forall \epsilon$, there exists a $N'$, such that for $n>N'$, $\mathbb{E}[P(\mathbf{Y}\in T_{\delta}(Y|\mathbf{x}))]_{\mathbf{x}}>1-\epsilon$. The nuance of this second result is such: as the length of the sequence becomes large, the probability of $\mathbf{x}$ being in the typical set tends to 1, as does the probability of $\mathbf{y}$ being in the typical set of the chosen $\mathbf{x}$. The two cases where a) $\mathbf{x}$ is not typical, b) $\mathbf{y}$ is not typical (w.r.t. $\mathbf{x}$) have vanishing probability.\\

With this, we have the tools to prove our desired result: the asymptotic Schmidt measure is given by the entanglement cost. We will use that the entanglement cost is equivalent to $\lim_{n\rightarrow \infty} E_F(\rho^{\otimes n})/n$, the regularised entanglement of formation. When calculating $E_F$, there exists an optimal decomposition whose number of pure states does not exceed $\dim(\mathcal{H})^{2n+1}$. This is a consequence of the continuity of the von Neumann entropy \cite{U2010}. We now consider this decomposition $\rho^{\otimes n} = \sum_k p_k \ket{\Phi_k}\bra{\Phi_k}$, where each pure state $\Phi_k$ has Schmidt decomposition $\Phi_k = \sum_{i} \lambda_{i|k} \ket{e_{i|k}f_{i|k}}$. Note that the set $\lambda_{i|k}^2$ form a probability distribution, and that $(k,i)$ form a joint probability distribution. This means, in particular, we can consider the (strongly) typical set $T_{\delta}(K,I)$. This is well defined as for any pure state, $I$ is also finite (its size is upper bounded by $\dim(\mathcal{H})^n$). 

We will now consider the regularised Schmidt measure $\lim_{m\rightarrow \infty} S_M[(\rho^{\otimes n})^{\otimes m}]/m$. We can upper bound $S_M[(\rho^{\otimes n})^{\otimes m}]$  by $\sum_{j} q_j \log S_R(\Psi_{j})$, where $\{q_j,\Psi_{j}\}$  is a specific decomposition of $(\rho^{\otimes n})^{\otimes m}$. We will choose the one generated when substituting the optimal (with respect to $E_F$) decomposition of $\rho^{\otimes n}$ into $(\rho^{\otimes n})^{\otimes m}$ and expanding out the tensor product; thus we equate $\{q_j\}=\{p_{k_1}p_{k_2}\ldots\}$ and $\{\Psi_{j}\}=\{\Phi_{k_1}\otimes \Phi_{k_2}\ldots\}$. To reflect this, we will subsequently write this decomposition as $(\rho^{\otimes n})^{\otimes m}=\sum_{\mathbf{k}}p_{\mathbf{k}}\ket{\Phi_{\mathbf{k}}}\bra{\Phi_{\mathbf{k}}}$.\\

We can define a new (subnormalised state) $\rho^{n,m} = \sum_{\mathbf{k}} p_{\mathbf{k}} \ket{\Phi'_{\mathbf{k}}}\bra{\Phi'_{\mathbf{k}}}$, where $\ket{\Phi'_{\mathbf{k}}}=\sum_{\mathbf{i}\in T_{\delta}(I|k)} \boldsymbol{\lambda}_{\bs{i}|\bs{k}}\ket{\boldsymbol{e_{i|k}f_{i|k}}}$, $\boldsymbol{\lambda}_{\bs{i}|\bs{k}}=\lambda_{i_1|k_1}\lambda_{i_2|k_2}\ldots$ and $\boldsymbol{e_{i|k}}=e_{i_1|k_1}\otimes e_{i_2|k_2}\ldots$.
The (one-norm) distance between this state and $(\rho^{\otimes n})^{\otimes m}$ is given by $\|(\rho^{\otimes n})^{\otimes m}- \rho^{n,m}\|_1 \leq \mathbb{E}[P(i\not \in T_{\delta}(I)]_{\mathbf{k}}$. 
Using the properties of strong typicality, we know that for any $\epsilon$, $\exists\, m$ sufficiently large such that the above expectation $<\epsilon$, and that $\lim_{m\rightarrow \infty} \Rightarrow \epsilon \rightarrow 0$. \\

In particular, that means that
$\lim_{m\rightarrow \infty} S_M^{\epsilon}((\rho^{\otimes n})^{\otimes m})/m \leq \lim_{m\rightarrow \infty} S_V(\rho^{n,m})/m$, where "$S_V$" stands for the ``Schmidt value", i.e., $\sum_\mathbf{k} p_\mathbf{k} \log S_{R}(\mathbf{\Phi'_k})$\ of this decomposition of mixed state $\rho^{n,m}$. We now note the Schmidt rank of $\ket{\mathbf{\Phi}'_\mathbf{k}}$ is equal to the number of non-zero sequences in $T_{\delta}(I)$ which, for large enough $m$, is upper bounded by $2^{m(H(I|K)(1+\delta))}$. However, $H(I|K)= \sum_{k}p_{\mathbf{k}}H(I|k) = E_{F}(\rho^{\otimes n})$, as we chose $\{p_\mathbf{k},\mathbf{\Phi_k}\}$ to be the optimal decomposition.\\

The consequence of this is that, for any $\delta>0$,
$\lim_{m\rightarrow \infty} S_M^{\epsilon}((\rho^{\otimes n})^{\otimes m})/m \leq E_F(\rho^{\otimes n})(1+\delta)$. We can divide this equation through by $n$, and then take the limit of $n\rightarrow \infty$, since this equation holds for any $n$. From this we obtain $\lim_{n\rightarrow \infty}\lim_{m\rightarrow \infty} S_M^{\epsilon}((\rho^{\otimes n})^{\otimes m})/mn \leq E_C(\rho)(1+\delta)$. As this true for any $\delta>0$, we can drop the $\delta$ term. However, the left hand side is equivalent to $\lim_{k \rightarrow \infty} S_M^{\epsilon}((\rho^{\otimes k}))/k$. This gives us the result that $\lim_{k \rightarrow \infty} S_M^{\epsilon}((\rho^{\otimes k}))/k \leq E_C(\rho)$.\\

To prove equality, we may note that given any specific decomposition, the associated entanglement quantity is lower than the associated Schmidt quantity. Therefore, we have that $E_F(\rho)\leq S_M(\rho)$ and similarly that $E_F^{\epsilon}\leq  S_M^{\epsilon}$. This means $E_F^{\epsilon}((\rho^{\otimes n})^{\otimes m})\leq S_M^{\epsilon}((\rho^{\otimes n})^{\otimes m})$.  Our strong typicality tells us in the limit of $m \rightarrow \infty$ that $\epsilon$ vanishes. Therefore the LOCC protocol given by $E_F^{\epsilon}(\rho^{\otimes n})$ is a valid protocol considered in the definition of $E_C$. Therefore $E_{C}(\rho) \leq \lim_{n\rightarrow \infty} \lim_{m\rightarrow \infty} E_F^{\epsilon}((\rho^{\otimes n})^{\otimes m})$, which is itself a lower bound for the regularisation of the smoothed Schmidt measure, $\lim_{k \rightarrow \infty} S_M^{\epsilon}((\rho^{\otimes k}))/k$. Since we have sandwiched the regularised Schmidt measure from above and below by the entanglement cost, the two must coincide.

\subsection{Asymptotic Schmidt measure for Steering}
We see in the previous section that, for states, the entanglement and dimensionality converge as the number of copies increases. One can ask if a similar phenomenon occurs for assemblages. To do this, we need a notion of entanglement for steering assemblages. This is given by the \emph{steering entanglement of formation} \cite{C2021},
\begin{equation}
E_{FA}(\bs{\sigma}):= \inf_{p_i,\bs{\tau_i}} \sum_i p_i S(\rho_{\bs{\tau_i}}).
\end{equation}
It has the interpretation as the minimal entanglement required to create the assemblage $\bs{\sigma}$ (as measured by the entanglement of formation). This quantity is continuous on the set of extremal assemblages (proof in Appendix \ref{SEFA}), meaning it also has a finite (and achievable) optimal decomposition \cite{U2010}.\\

We would like to consider the regularised smoothed Schmidt measure for assemblages, $\lim_{n\rightarrow \infty} S_M^{\epsilon}(\bs{\sigma}^{\otimes n})/n$; to do that, we need a definition of multiple copies of an assemblage. $n$ copies of an assemblage is denoted by $\bs{\sigma}^{\otimes n}$, whose components are given by $\sigma^{\otimes n}_{a_1\ldots a_n|x_1\ldots x_n}=\sigma_{a_1|x_1}\otimes \ldots \sigma_{a_n|x_n}$. The smoothed version of the Schmidt measure is given  by:
\begin{equation}
S_M^{\epsilon}(\bs{\sigma}):=\inf_{\|\bs{\sigma}-\bs{\sigma}'\|_1\leq \epsilon} S_M(\bs{\sigma}'),
\end{equation}
where we choose the trace norm from \cite{KCBMCN2018}:
$\|\bs{\sigma}-\bs{\tau}\|_1= \sup_{P_X} \sum_{a,x}p(x)\|\sigma_{a|x}-\tau_{a|x}\|_{1}$.
There are alternative choices to generalise the trace norm to assemblages, but this choice does not change under copying of an input, or blow up when considering multiple copies of the assemblages. Furthermore, it fits well with the operational meaning of the original trace norm, as a measure of distinguishability.\\

Satisfyingly, for steering assemblages we can again find a link between dimensionality and entanglement; the regularisation of the Schmidt measure satisfies
\begin{equation}
\lim_{k\rightarrow \infty}E_{\mrm{FA}}^{\epsilon}(\bs{\sigma}^{\otimes k})/k \leq \lim_{k\rightarrow \infty} S_{M}^{\epsilon}(\bs{\sigma}^{\otimes k})/k\leq  \lim_{k\rightarrow \infty}E_{\mrm{FA}}(\bs{\sigma}^{\otimes k})/k,
\end{equation}
with $\epsilon$ vanishing at $k\rightarrow \infty$. The proof of this result follows a similar pattern as for the analysis for quantum states: we consider several copies of an assemblage, $\bs{\sigma}^{\otimes n}$. The entanglement of formation for this assemblage is $E_{\mrm{FA}}(\bs{\sigma}^{\otimes n})$. We can denote the optimal decomposition for $E_{\mrm{FA}}(\bs{\sigma}^{\otimes n})$ by $\{p_k,\bs{\tau}_k\}$. This means we have an explicit decomposition of $(\bs{\sigma}^{\otimes n})^{\otimes m} = \sum_{\mathbf{k}} p_{\bf{k}}\bs{\tau}_{\bf{k}}$. For any assemblage in our decomposition we can create an explicit construction of it. The marginal of $\bs{\tau}_{\bs{k}}$ is $\rho_{\bs{\tau}_{\bs{k}}}=\rho_{\tau_{k_1}}\otimes\rho_{\tau_{k_2}}\ldots \rho_{\tau_{k_m}}\equiv \sum_{\mathbf{i}} \lambda^2_{\mathbf{i}|\mathbf{k}} \ket{e_{\mathbf{i}|\bf{k}}}\bra{e_{\mathbf{i}|\bf{k}}}$. Using the purification of this state $\ket{\Phi_{\bf{k}}}=\sum_{\mathbf{i}}\lambda_{\bf{i}|\bf{k}}\ket{e_{\mathbf{i}|\bf{k}}}\ket{e_{\mathbf{i}|\bf{k}}}$ and measurements $M_{\bf{a}|\bf{x}}=\rho^{-1/2}_{\bs{\tau}_{\bs{k}}}\tau^{T}_{\bs{a}|\bf{x}}\rho^{-1/2}_{\bs{\tau}_{\bs{k}}}$, we obtain a realisation of this assemblage.\\

For each $\bs{k}$ we can create a new (subnormalised) assemblage by replacing the state $\ket{\Phi_{\bf{k}}}$ with the state $\ket{\Phi'_{\mathbf{k}}}$, where $\ket{\Phi'_{\mathbf{k}}}=\sum_{\mathbf{i}\in T_{\delta}(I|k)} \boldsymbol{\lambda}_{\bs{i}|\bs{k}}\ket{\boldsymbol{e_{i|k}f_{i|k}}}$. We write this new assemblage as $\bs{\tau}'_{\bs{k}}$. We can mix these assemblages together with appropriate weights to obtain $\bs{\sigma}^{n,m}= \sum_{\bs{k}} p_{\bs{k}}\bs{\tau}'_{\bs{k}}$, which approximates our original assemblage. The distance between our approximation and original is given by $\|\left(\bs{\sigma}^{\otimes n}\right)^{\otimes m} - \bs{\sigma}^{n,m}\|_1 \leq\mathbb{E}[P(i\not \in T_{\delta}(I)]_{\mathbf{k}}$. \\

Strong typicality again tells us that this expectation can be made smaller than any $\epsilon$ for large enough $m$, and that it disappears in the limit $m\rightarrow \infty$. From this argument, we can conclude that $\lim_{m\rightarrow \infty} S^{\epsilon}_{M}(\left(\bs{\sigma}^{\otimes n}\right)^{\otimes m})\leq \lim_{m\rightarrow \infty} S_V(\bs{\sigma}^{n,m})/m$, where $S_V$ is the ``Schmidt value" of the chosen decomposition of assemblage $\bs{\sigma}^{n,m}$. The Schmidt value of this particular decomposition is $\sum_{\bs{k}} p_{\bs{k}} \log \mathrm{rank}(\rho_{\bs{\sigma}'_{\bs{k}}}) \equiv \sum_{\bs{k}} p_{\bs{k}} \log S_R(\ket{\Phi'_{\bs{k}}})$. As we have already seen, the Schmidt rank is upper bounded for large enough $m$ by $2^{m(H(I|K)(1+\delta))}$. We now use  $H(I|K)= \sum_{k} p_{k} H(I|k) = E_{FA}(\bs{\sigma}^{\otimes n})$, where the final equality comes from our choice $\{p_k, \bs{\sigma}_k\}$ of the optimal decomposition.\\

This allows to conclude that for any $\delta>0$, $\lim_{m\rightarrow \infty} M_{S}^{\epsilon}((\bs{\sigma}^{\otimes n})^{\otimes m})/m \leq E_{FA}(\bs{\sigma}^{\otimes n})(1+\delta)$. We can therefore drop the $\delta$. Finally, we divide both sides by $n$, and take the limit. This gives the inequality:
$\lim_{n\rightarrow \infty}\lim_{m\rightarrow \infty} S_M^{\epsilon}((\bs{\sigma}^{\otimes n})^{\otimes m})/mn \leq \lim_{n\rightarrow \infty} E_{FA}(\bs{\sigma^{\otimes n}})/n$, which is equivalent to 
$\lim_{k\rightarrow \infty} S_M^{\epsilon}(\bs{\sigma}^{\otimes k})/k \leq \lim_{k\rightarrow \infty} E_{FA}(\bs{\sigma^{\otimes k}})/k$.\\

We can similarly provide a lower bound in the following form: As for any assemblage $E_{FA}(\bs{\sigma})\leq S_{M}(\bs{\sigma})$ (simply via $S(\rho)\leq \mathrm{rank}(\rho)$), then $E_{F}^{\epsilon}(\bs{\sigma}^{\otimes k})/k\leq S_{M}^{\epsilon}(\bs{\sigma}^{\otimes k})/k$ and thus too $\lim_{k\rightarrow \infty} E_{F}^{\epsilon}(\bs{\sigma}^{\otimes k})/k\leq \lim_{k\rightarrow \infty} S_{M}^{\epsilon}(\bs{\sigma}^{\otimes k})/k$.\\

We note that one cannot take this proof one step further, and equate the two sides of the inequality with the entanglement cost for steering assemblages \cite{C2021}. The key property missing for this is the continuity of $E_{FA}$ over \emph{all} assemblages, which does not hold. Furthermore, we are unsure if a gap can exist between the upper and lower bound. Nevertheless, we see that the entangled degrees of freedom required to create an assemblage, and the entanglement itself converge in the limit of many copies.\\

We would like to comment on the asymptotic treatment of the dimension measure; especially with regards to measurement sets. Intuitively, a smoothed version would give the memory cost of a set of measurements subject to a small error - which would vanish in the asymptotic limit. This is exactly in the spirit of Shannon's source coding theorem. We therefore believe it would be of interest to analyse the asymptotic behaviour of the dimension measure; and this would lead to compression limits on incompatible measurements. Unlike the cases analysed here, we do not have a notion of entanglement for measurements, so it is an open question whether this analysis would converge to a previously studied property.

\section{Conclusions}

We introduced the concept of average dimensionality for three different types of quantum devices. These consist of quantum channels, sets of quantum measurements, and steering assemblages. For channels, the measure characterises the coherence that a channel preserves on average. For the case of measurements, our concept has an interpretation as the average compression dimension that the involved measurements allow. In the case of steering, the measure represents the average entanglement dimensionality in terms of the Schmidt measure, i.e. the number of degrees of freedom or levels that are entangled on average, that can be verified through a steering experiment.

We have shown that for channels between qubit and qutrit systems, the measure can be evaluated with semi-definite programming techniques. This is also the case for qubit measurements. In the case of bipartite quantum states, we first showed that their asymptotic Schmidt measure is equal to the entanglement cost, and then related our measure and entanglement of formation of steering assemblages in the asymptotic setting.

We believe that with the advent of high-dimensional quantum information processing, quantifiers of dimensionality akin to ours will be central for the benchmarking and verification of devices used in quantum communication protocols. Hence, for future works, it will be crucial to analyse the concept of average dimensionality for the three types of quantum devices discussed here, as well as to develop the theory for more general scenarios, such as ones involving quantum instruments, quantum processes and continuous variable systems.

\section{Acknowledgements}
We would like to thank René Schwonnek, Tobias Osborne, Sébastien Designolle, Benjamin Jones, Otfried Gühne, Nicolas Brunner, Jonathan Steinberg, and Chau Nguyen for useful discussions.
T.C. would like to acknowledge the funding Deutsche Forschungsgemeinschaft (DFG, German Research Foundation) under Germany's
Excellence Strategy EXC-2123 QuantumFrontiers 390837967, as well as the support of the Quantum Valley Lower Saxony and the DFG through SFB 1227 (DQ-mat). R.U. is grateful for the support from the Swiss National Science Foundation (Ambizione PZ00P2-202179).

\bibliography{SMIRef.bib}
\onecolumngrid
\appendix

\captionsetup{justification=centering}

\newpage

\section{Table of Terms}
In this appendix, we provide a table summarising the notation, for ease of the reader.\\

\renewcommand{\arraystretch}{1.4}
\begin{table}[h!]
    \centering
    \begin{tabular}{ccc}
        Notation & Name & Definition \\
         \hline
         CPTP & \bf{C}ompletely \bf{P}ositive \bf{T}race \bf{P}reserving & \\
         LOCC & \bf{L}ocal \bf{O}perations and \bf{C}lassical \bf{C}ommunication &\\
         MUB & \bf{M}utually \bf{U}nbiased\bf{B}ases & \\
         n-PEB & n \bf{P}artially \bf{E}ntanglement \bf{B}reaking &\\
         POVM & \bf{P}ositive \bf{O}perator \bf{V}alued \bf{M}easure &\\
         SEP & The set of \bf{Sep}arable states & \\
         SDP & \bf{S}emi-\bf{D}efinite \bf{P}rogramming.\\
         \hline 
        
         $S_R(\phi)$ & Schmidt Rank & $\min r\; \text{s.t.} \ket{\phi}=\sum_{k=1}^r \lambda_k \ket{e_k}\ket{f_k}$\\
         
         $S_N(\rho)$ & Schmidt Number (states) & $\min \{\max_i \log S_R(\phi_i)\mid \sum_{i} p_i\ket{\phi_i}\bra{\phi_i} = \rho\}$\\
         
         $\mathbb{I}_d$ & Identity Operator on $\mathcal{H}_d$ &\\
         
         $S_N(\mathcal{E})$ & Schmidt Number (channels) & $\min \{\max_i \log \mathrm{rank}(K_i)\mid \mathcal{E}(\rho)=\sum_{i} K_i\rho K_i^{\dagger}\}$\\
         
         $S_M(\rho)$ & Schmidt measure (states) & $\min \{\sum_i p_i \log S_R(\phi_i)\mid \sum_{i} p_i\ket{\phi_i}\bra{\phi_i} = \rho\}$\\
         
         $S_M(\bs{\sigma})$ & Schmidt measure (assemblages) & $\min \{\sum_i p_i \log \rank(\rho_{\bs{\tau}^i})\mid \sum_{i} p_i\bs{\tau}^i = \bs{\sigma}\}$\\
         
         $D_M(\mathcal{E})$ & Dimension Measure (channels) & $\min \{\max_\rho \sum_i \Tr(K_i^\dagger K_i \rho)\log \rank K_i \mid \mathcal{E}(\rho)=\sum_{i} K_i\rho K_i^{\dagger}\}$\\
         
        $D_M(\bs{M})$ & Dimension Measure (measurements) & $\min \{\max_\rho \sum_i \Tr(K_i^\dagger K_i \rho)\log \rank K_i \mid \exists \bs{N}\, \mathcal{E}^*(N_{a|x})=M_{a|x}\}$\\
        
        $\mathcal{E}_p(\rho)$ & Qubit Depolarising Channel & $ (1-p)\rho + p\mathbb{I}_2/2$\\
        
        $\chi_{\mathcal{E}}$ & Choi Matrix of a channel. & $\mathbb{I}\otimes \mathcal{E}\left(\ket{\Phi^+_d}\bra{\Phi^+_d}\right)$\\
        
        $R_q(\rho)$ & Qubit Erasure Channel & $ (1-q)\rho + q \ket{e}\bra{e}$, $\braket{e}{0} = 0 = \braket{e}{1}$\\
        
        $\bs{M}$ & A collection of POVMs & $\{M_{a|x} \mid \sum_{a} M_{a|x} = \mathbb{I}_d\}$\\
        
        $D_N(\bs{M})$ & Compression Dimension & $\min \{ \max_i \log \rank K_i \mid \exists \bs{N}\, \mathcal{E}^*(N_{a|x})=M_{a|x}\}$\\
        
        $\mathcal{B}_{\mathrm{JM}}$ & Set of Jointly Measureable Pseudo-Measurements & $\{\mathbf{B}\mid B_{a|x}=\sum_{\lambda} D(a|x,\lambda)G_{\lambda},\;\sum_{\lambda}G_{\lambda}\leq \mathbb{I}\}$\\
        
        $S_M^{\varepsilon}(\rho)$ & Smoothed Schmidt measure (states) & $\inf_{\|\rho-\rho'\|_1\leq \epsilon} S_M(\rho')$. \\
        
        $S_M^{\varepsilon}(\rho)$ & Smoothed Schmidt measure (assemblages) & $\inf_{\|\bs{\sigma}-\bs{\sigma}'\|_1\leq \epsilon} S_M(\bs{\sigma}')$. \\
        
        $E_C(\rho)$ & Entanglement Cost (state) & $\inf\{E\mid \forall \epsilon,\;\; \delta >0\;\; \exists m,\,n,\, \Lambda,$\\
        & & 
$| E-m/n|\leq \delta,\;\; \|\Lambda(\ket{\Phi^+}\bra{\Phi^+}^{\otimes m}),\rho^{\otimes n}\|_{tr} \leq \epsilon\}$\\

$E_F(\rho)$ & Entanglement of Formation (States) & 
$\inf \{ \sum_i p_i E_S(\Phi_i) \mid \sum_{i} p_i\ket{\phi_i}\bra{\phi_i} = \rho\}$\\

$E_S(\phi)$ & Entanglement Entropy & $S(\mathrm{Tr}_{B}[\ket{\Phi}\bra{\Phi}])$\\ 

$T_{\delta}(X)$ & Strong Typical Set & $\{\mathbf{x}\bs{\mid} |N[x|\mathbf{x}]/n-P(x)|\leq \delta P(x)\}$ \\

$T_{\delta}(X,Y)$ & Strong Jointly Typical Set & $\{(\mathbf{x},\mathbf{y})\mid |N[(x,y)|(\mathbf{x},\mathbf{y})]/n-P(x,y)|\leq \delta P(x,y)$ \\

$T_{\delta}(Y|\mathbf{x})$ & Strong Conditionally Typical Set & $\{\mathbf{y}\mid (\mathbf{x},\mathbf{y})\in T_{\delta}(X,Y)\}$\\

$N[x|\mathbf{x}]$ & Number of $[\mathbf{x}]_i = x$ & \\

$H(X)$ & Shannon Entropy & $-\sum_{x} p_x \log p_x $\\

$H(Y|X)$ & Conditional Shannon Entropy & $-\sum_{x} p(x)\sum_{y} p(y|x) \log p(y|x)$ \\

$S_V(\{(p_i,\phi_i)\}$ & ``Schmidt Value" & $\sum_i p_i \log S_R(\phi_i)$\\

$E_F^{\varepsilon}(\rho)$ & Smoothed Entanglement of Formation (states) & $\inf_{\|\rho-\rho'\|_1\leq \epsilon} E_F(\rho')$. \\

$S(\rho)$ & von Neumann Entropy & $-\mathrm{Tr}(\rho\log\rho)$\\

$E_{FA}(\bs{\sigma})$ & Steering Entanglement of Formation (Assemblages) & $\inf \{\sum_i p_i S(\rho_{\bs{\tau}^i})\mid \sum_{i} p_i\bs{\tau}^i = \bs{\sigma}\}$ \\

$E^{\varepsilon}_{FA}(\bs{\sigma})$ & Smoothed Steering Entanglement of Formation & $\inf_{\|\bs{\sigma}-\bs{\sigma}'\|_1\leq \epsilon} E_{FA}(\bs{\sigma}')$\\

$S_{M}^n$ & ``Restricted" Schmidt measure & $\min \{\sum^{n}_i p_i \log S_R(\phi_i)\mid \sum_{i} p_i\ket{\phi_i}\bra{\phi_i} = \rho\}$\\

$W(\bs{M})$ & Incompatibility Weight & $\min \{p \bs{M} = p\bs{N}+(1-p)\bs{G}, \bs{G} \text{ jointly measureable}\}$

    \end{tabular}
    \label{tab:my_table}
\end{table}

\newpage

\section{Finiteness and Achievability of Schmidt and Dimension Measures}
One may note that the Schmidt and dimension measures introduced in the main body of the paper are defined via minimisations and maximisations, rather than infimums and supremums. In this section we will justify this choice; by showing that for these measures the two definitions are equivalent. We will also show, as a consequence, that the number of terms in the optimal decomposition is \emph{finite}. We will also see that all of these measures are lower semi-continuous.\\
\subsection{Schmidt measure for States}
We will begin with a detailed proof for the Schmidt measure (for states); and then adjust this argument, as required, to our other measures. Suppose we instead define the Schmidt measure as:
\begin{equation}
S_M(\rho) = \inf \left\{\sum_i p_i \log S_R(\ket{\phi_i}\bra{\phi_i}) \bigg | \sum_i p_i \ket{\phi_i}\bra{\phi_i} =\rho\right\}.
\end{equation}
where the optimisation is over the set of all (possibly infinite length) convex combinations $\sum_i p_i\ket{\phi_i}\bra{\phi_i} = \rho$. We start by noting that $S_M(\rho)$ may be defined by the limit $S_M(\rho)=\lim_{n\rightarrow \infty} S_M^{n}(\rho)$, where $S_M^{n}(\rho)$ is defined as above, except that the number of (non-zero) terms in the decomposition must be no larger than $n$. Thus the infimum in the definition of $S_M^{n}(\rho)$ is over a compact set. The function in question, $\sum^{n}_i p_i \log S_R(\ket{\phi_i}\bra{\phi_i})$, is a \emph{lower semi-continuous} function, because the Schmidt rank is lower semi-continuous, and the logarithm function is monotonically increasing - a property which preserves lower semi-continuity under composition. $p_i$ is linear, and thus the product $p_i \log S_R(\ket{\phi_i}\bra{\phi_i})$ is also lower semi-continuous. Moreover, the sum of lower semi-continuous functions is again lower semi-continuous, proving our claim. The infimum of an lower semi-continuous function over a compact set is achieved and therefore a minimum; telling us that we may use a minimisation in the definition of $S_M^n(\rho)$. For a detailed description about lower semi-continuous functions and their properties, we refer the reader to \cite{J2003,M1999}\\

This tells that, in order for $S_M(\rho)$ to be a true infimum, we must have an infinite sequence of $n_k$, for which $S_M^{n_{k}+1}(\rho)<S_M^{n_k}(\rho)$. 

Let us pick one such $n_k=:\tilde{n}\geq d^4 $. Let us consider the optimal decomposition for $n_{k}+1$, $\{\ket{\phi_i}\bra{\phi_i}\}$. This set must be linearly dependent. Let us choose $\ket{\phi_\alpha}$ as the state which has the highest Schmidt rank and can be written with a nonzero coefficient in a linear dependence equation. Using the linear dependence equation, we can write a new decomposition of $\rho = \sum_{i}^{n_k+1} p_i' \ket{\phi_i}\bra{\phi_i}$ in which one of the $p_i'$ is zero; that is, it is a valid decomposition (and thus provides an upper bound) for $S_M^{n_k}$. However, we have chosen this decomposition such that $S_M^{n_k}(\rho)\leq \sum_{i}^{n_k+1} p_i' \log S_{R}(\ket{\phi_i}) \leq \sum_{i}^{n_k+1} p_i \log S_{R}(\ket{\phi_i}) = S_M^{n_k+1}(\rho)$. However, this contradicts our assumption, implying it cannot hold. Thus the infimum in the definition of $S_M(\rho)$ must also be a minimum.\\

As a consequence of this result, we see that the Schmidt measure must be achieved by at most $d^4$ pure states. We also obtain as a consequence that the Schmidt measure \emph{itself} is a lower semi-continuous function:  $S_M^{n}$ is the minimum of a set of lower semi-continuous functions, which is again lower semi-continuous. However, this means that $S_M(\rho)$ is also a minimum over a set of lower semi-continuous functions, implying it, too, is lower semi-continuous.\\

Before we move on, it is also worth mentioning here the possibility of replacing the sum in the definition with an integral over the pure states instead. Fortunately, assuming that the state of interest is finite-dimensional, this is not necessary. Suppose we had such an optimal integral $\int_{\Omega} \rho(\omega) \ket{\omega}\bra{\omega} \log S_R(\ket{\omega})$, such that $\int_{\Omega} \rho(\omega) \ket{\omega}\bra{\omega}=\rho$. Let us evaluate the restricted sub-integral $\int_{\Omega|\log S_R(\ket{\omega})=0} \rho(\omega) \ket{\omega}\bra{\omega}=:\rho_0$. We can define an integral version of the Schmidt number for $\hat{\rho}_0:=\rho_0/\mathrm{Tr}(\rho_0)$, as $S^{\infty}_{N}(\rho_0)=\inf_{\rho(\omega)} \sup_{\omega}\{\log S_R(\ket{\omega}) |  \int_{\Omega} \rho(\omega) \ket{\omega}\bra{\omega}=\hat{\rho}_0\}$.  It is clear this quantity is $0$. Moreover, $\sup_{\omega}\{\log S_R(\ket{\omega})\}$ is a lower semi-continous function, implying the set $\{\sigma| S^{\infty}_{N}(\sigma)\leq 0\}$ is closed (and thus compact). Catheodory's theorem tells us we can replace then our subintegral for a sum, which will also give $\rho_0$. 
Let us now consider $\int_{\Omega|\log S_R(\ket{\omega})=1} \rho(\omega) \ket{\omega}\bra{\omega}=:\rho_1$. Since $\{\sigma| S^{\infty}_{N}(\sigma)\leq 1\}$ is also compact, we can replace this second integral by a sum too, into extremals whose logarithmic Schmidt rank is $0$ or $1$. If any extremal with non-trivial weight were to have rank $0$, however, this would strictly decrease our Schmidt measure; a contradiction to optimality. Thus we may replace our integral with a sum of extremals with logarithmic Schmidt rank 1. Repeating this argument inductively gives us a finite sum decomposition of $\rho$, which achieves the same Schmidt measure as our integral. Thus, we need not consider integrals in our definition. This same argument works for the other measures introduced in this work.
\subsection{Schmidt measure for Assemblages}
If we instead define
\begin{equation}
S_M(\bs{\sigma}) = \inf \left\{\sum_i p_i\log \mathrm{rank}\, \rho_{\bs{\tau}^i} \bigg | \sum_i p_i \bs{\tau}^i =\bs{\sigma}\right\}
\end{equation}
we see that an argument similar to the above is almost directly applicable. Instead of Schmidt rank, we have $\mathrm{rank} \,\rho_{\bs{\tau}}$ - however, matrix rank is also a lower semi-continuous function (one can use this fact to prove lower semi-continuity of the Schmidt rank!). The other difference to take care of is to check that convex combinations of $n$ assemblages summing to $\bs{\sigma}$ is a compact set - since we are assuming that $\mathcal{A},\mathcal{X}$ are finite, and the characterised Hilbert space is finite, we see that this space does have a maximal dimension and is thus compact. For the proof by contradiction, we would then take $n_k$ to be greater than this dimension.\\

A consequence of this result is that the Schmidt measure for assemblages is lower semi-continuous. To work out the maximal number of elements in the optimal decomposition, we must obtain the dimension of the space of $d$- dimensional, $\mathcal{X}$-input, $\mathcal{A}$-output assemblages. We have $|\mathcal{X}||\mathcal{A}|$ subnormalised states of dimension $d^2$ - however $|\mathcal{X}|$ of these are fixed by the no-signalling condition. We also have the marginal state, which has $d^2-1$ degrees of freedom, as it must be normalised. This means the dimension of our space is $\mathrm{dim}=d^2(|\mathcal{X}|(|\mathcal{A}|-1)+1)-1$. This means our optimal decomposition can have at most $\mrm{dim}+1=d^2\left(|\mathcal{X}|(|\mathcal{A}|-1)+1\right)$ components.

\subsection{Dimension Measure for Measurements}\label{DimMeasMeas}
The more general definition of this measure would be:
\begin{equation}
D_{M}(\bs{M})=\inf_{\bs{A}^i} \sup_{\rho} \sum_{i} \mathrm{Tr}(A^i\rho)\log \mrm{rank}(A^i).
\end{equation}
Let us first fix the $\bs{A}^{i}$, and consider the supremum over states $\rho$. $\sum_{i} \mathrm{Tr}(A^i\rho)\log \mrm{rank}(A^i)$ is a linear function over $\rho$ - and our supremum is taken over the set of density matrices for a (finite) input dimension. This is the supremum of a linear function over a compact set - meaning it is achieved, and can be replaced by a maximisation.\\

We can also note that, for a fixed $\rho$, $\mathrm{Tr}(A^i\rho)$ is a linear function of $\bs{A}^i$. We have already seen the $\log \mrm{rank}$ is a lower semi-continuous function, and that multiplying it with a linear function and performing a summation preserves lower semi-continuity. Although we perform a maximisation over $\rho$, a maximisation over lower semi-continuous functions also preserves lower semi-continuity. This means we are now in a position to repeat the same argument we used for the Schmidt measure - provided that a finite number of $\bs{A}^{i}$ is a compact set. However, each $\bs{A}^i$ is a collection of $|\mathcal{X}||\mathcal{A}|$ Hermitian matrices (the set of which has dimension $d^2$); so, as with assemblages, we can conclude that a finite number of $\bs{A}^i$ is a compact space.\\

Another small variation in the previously used argument is that, previously, we decomposed our state/assemblage as a convex combination, whereas here we write directly $\bs{M} = \sum_{i} \bs{A}^i$. However, this can be rewritten as a convex combination, by rescaling each $\bs{\hat{A}}^{i}= \bs{A}^i$. Then $\bs{M}=\sum_{i} p_i \bs{\hat{A}}^{i}$, where $p_i$ is simply the inverse scaling factor, $\mathrm{Tr}(A^i)/d$. One should take care to note however, that in general $\hat{A}^{i}\not\leq  \openone_d$ - these objects are not physical. Nevertheless, we can use them to directly repeat the argument used in the proof for the Schmidt measure. \\

A consequence of this unphysical representation is that we can use it determine the maximal number of components required in the optimal decomposition. As is the case with assemblages, we have $|\mathcal{X}||\mathcal{A}|$ Hermitian operators, $|\mathcal{A}|$ of which are determined by no-signalling conditions; as well as a marginal operator with a fixed trace - thus giving it $d^2-1$ degrees of freedom. We therefore have a $\mathrm{dim}=d^2(|\mathcal{X}|(|\mathcal{A}|-1)+1)-1$ space, with optimal decompositions achievable with $d^2(|\mathcal{X}|(|\mathcal{A}|-1)+1)$ components.

\subsection{Dimension Measure for Channels}
If one tries to directly repeat the same proof for the general definition of dimension measure for the channel,
\begin{equation}
D_M(\mathcal{E}):=\min_{K_i} \max_{\rho}\mathrm{Tr}(K_i^{\dagger}K_i\rho) \log \mathrm{rank}(K_i),
\end{equation}
one runs into a problem; the set of Kraus representations for a given channel is not a convex set. This means that one cannot re-express linearly dependent operators $K_{i}^\dagger K_i$ and obtain a valid decomposition with a reduced number of terms.\\

To solve this problem, we use an equivalent formulation of the dimension measure, that we used to  express the dimension measure for low dimensions as an SDP. That is, we can write it as:
\begin{equation}
D_M(\mathcal{E}):=
\inf_{\phi^l} \sup_\rho \sum_l d\,\mathrm{Tr}(\rho^{T}\otimes \openone\ket{\phi^l}\bra{\phi^l})\log S_R(\ket{\phi^l}).
\end{equation}
Now we can follow the same lines as before: $\sum_i d\,\mathrm{Tr}(\rho^{T}\otimes \openone\ket{\phi^l}\bra{\phi^l})\log S_R(\ket{\phi^l})$ is a linear function of $\rho$, allowing us to replace the supremum with a maximimum.  Similarly to before, $\max_{\rho} \sum_i d\,\mathrm{Tr}(\rho^{T}\otimes \openone\ket{\phi^l}\bra{\phi^l})\log S_R(\ket{\phi^l}$ is a lower semi-continuous function. By normalising each $\ket{\hat{\phi}^l}\bra{\hat{\phi}^l}=\ket{\phi^l}\bra{\phi^l}/\braket{\phi^l}{\phi^l}$ our infimum is performed over the convex set $\sum_{i} p_i\ket{\hat{\phi}^l}\bra{\hat{\phi}^l} = \chi_{\mathcal{E}}$, allowing us to use the direct finite achievability argument first applied to the Schmidt measure.\\

As with the Schmidt measure, we obtain as a consequence of this argument that an optimal decomposition is achievable with $d^4$ terms; and that, as with the other Schmidt and dimension measures, this measure is lower semi-continuous.

\section{Properties of the Channel Schmidt measure}
In the following section, we prove some results regarding the Schmidt number and meausure of a quantum channel, which are either stated or used in the main paper.
\begin{lemma}
The Schmidt number/Schmidt measure of a channel is equal to the Schmidt number/Schmidt measure of its Choi state.
\end{lemma}
\begin{proof}
The Choi state of a channel is given by $id\otimes\mathcal{E}(\ket{\Phi^+_d}\bra{\Phi^+_d})$, where $d$ in the input space dimension. We will omit the $d$ subscripts hereafter. Suppose we have an explicit decomposition of $\mathcal{E}$ into Kraus operators $\{K_i\}$. Then we can explicitly write our Choi state as:
\begin{equation}
\sum_{i} \left(\openone\otimes K_{i}\ket{\Phi^+}\right)\left(\openone\otimes K_{i}\ket{\Phi^+}\right)^{\dagger}.
\end{equation}
As $K_{i}$ is rank $r$, it can be expressed as the matrix $K_{i}=\sum_{j}^{r_i}\lambda^i_{j}\ket{e^i_j}\bra{f^i_j}$ where $e^i_j,f^i_j$ are left/right eigenvectors.\\
Thus
\begin{align*}
&\sum_{i} \left(\openone\otimes K_{i}\ket{\Phi^+}\right)\left(\openone\otimes K_{i}\ket{\Phi^+}\right)^{\dagger}\\
&= 1/d\sum_{i,j,k}\lambda^{i}_{j}\ket{k}\ket{e^i_j}\bra{f^i_j}\ket{k}\sum_{i',j',k'}\lambda^{i'}_{j'}\bra{k'}\bra{e^{i'}_{j'}}\bra{k'}\ket{f^{i'}_{j'}}\\
&=\left(1/\sqrt{d}\sum_{i,j}\lambda^{i}_{j}\ket{f^{*i}_j}\ket{e^i_j}\right)\left(1/\sqrt{d}\sum_{i',j'}\lambda^{i'}_{j'}\bra{f^{*i'}_{j'}}\bra{e^{i'}_{j'}}\right)\\
\end{align*}
where conjugation is taken with respect to the chosen basis for $\ket{\Phi^+}$. From this, we see that every Kraus decomposition corresponds to a decomposition of the Choi state into pure states, whose Schmidt rank is equal to the rank of the corresponding Kraus state. the respective weight of each pure state in the sum is given by $p_{i}=\sum^{r_i}_{j}(\lambda^{i}_{j})^2/d \equiv \mathrm{Tr}(K^\dagger_{i}K_{i})/d$.

Conversely, given a decomposition of the Choi state into pure states $\{p_{i},\sum_{j} \lambda^{i}_{j}\ket{f^j_{i}}\ket{e^{j}_{i}}\}$, we can construct a set of Kraus operators $K_{i}=\sqrt{dp_i}\sum_{j}\lambda^{i}_{j}\ket{e^i_{j}}\bra{f^{i}_{j}}$. One can check that these satisfy the Kraus requirements, the rank is equal to the Schmidt rank, and $p_{i}=\mathrm{Tr}(K^\dagger_{i}K_{i})/d$.
\end{proof}

\begin{corollary}
For channels with a single Kraus operator (i.e. unitary or isometric evolution), the Schmidt number and measure are both given by the log rank of this Kraus operator.
\end{corollary}
This simply follows from the result that such channels correspond to pure states, upon which the logarithmic Schmidt rank, measure and number all coincide.

\section{Properties of the Channel Dimension Measure}
In the following section, we verify the claims made in the main paper regarding the properties of our new measure. We remind the reader that the the dimension measure of a channel is defined as:
\begin{equation}
D_M(\mathcal{E}):=\min_{K_i} \max_{\rho}\mathrm{Tr}(K_i^{\dagger}K_i\rho) \log \mathrm{rank}(K_i).
\end{equation}
\begin{itemize}
\item Convexity: Suppose we have a channel $\mathcal{E}=p\mathcal{E}_1+(1-p)\mathcal{E}_2$. Suppose the optimal Kraus operators for $\mathcal{E}_1$ are $\{K_i\}$, and for $\mathcal{E}_2$ they are $\{L_j\}$. Then a valid choice of Kraus operators for $\mathcal{E}$ is $\{\sqrt{p}K_i\}\cup \{\sqrt{1-p}L_j\}$. Thus we have that:
\begin{align*}
D_M(\mathcal{E}) &\leq \max_{\rho} \sum_{i} \mrm{Tr}(pK_{i}^{\dagger}K_i\rho)\log \rank (pK_i) + \sum_{j} \mrm{Tr}((1-p)L_{j}^{\dagger}L_j\rho)\log \rank ((1-p)L_j)\\
&\leq p\max_{\rho_1}\mrm{Tr}(K_{i}^{\dagger}K_i\rho_1)\log \rank (K_i) + (1-p)\max_{\rho_2}\sum_{j} \mrm{Tr}(L_{j}^{\dagger}L_j\rho_2)\log \rank (L_j)\\
&=pD_M(\mathcal{E}_1)+(1-p)D_M(\mathcal{E}_2).
\end{align*}
\item Subadditivity over tensor product. Suppose we have a channel $\mathcal{E}=\mathcal{E}_1\otimes \mathcal{E}_2$. Suppose the optimal Kraus operators for $\mathcal{E}_1$ are $\{K_i\}$, and for $\mathcal{E}_2$ $\{L_j\}$. Then a valid choice of Kraus operators for $\mathcal{E}$ is $\{K_i\otimes L_j\}$. This means we have:
\begin{align*}
D_M(\mathcal{E}) &\leq \max_{\rho} \sum_{i,j} \mathrm{Tr}[(K_i^\dagger K_{i}\otimes L_j^\dagger L_j)\rho]\log\rank(K_i\otimes L_j)\\
&= \max_{\rho} \sum_{i,j} \mathrm{Tr}[(K_i^\dagger K_{i}\otimes L_j^\dagger L_j)\rho](\log \rank(K_i)+\log\rank(L_j))\\
&=\max_\rho \sum_{i}\mathrm{Tr}[(K_i^\dagger K_{i}\otimes\openone)\rho]\log\rank(K_i) + \sum_{j}\mathrm{Tr}[(\openone\otimes L_j^\dagger L_j)\rho]\log\rank(L_j)\\
&=\max_\rho \sum_{i}\mathrm{Tr}[K_i^\dagger K_{i}\rho_A]\log\rank(K_i) + \sum_{j}\mathrm{Tr}[L_j^\dagger L_j\rho_B]\log\rank(L_j)\\
&\leq \max_{\rho_{A}} \sum_{i}\mathrm{Tr}[K_i^\dagger K_{i}\rho_A]\log\rank(K_i) + \max_{\rho_{B}} \sum_{j}\mathrm{Tr}[ L_j^\dagger L_j\rho_B]\log\rank(L_j) = D_M(\mathcal{E}_1)+D_M(\mathcal{E}_2).
\end{align*}
\item Non-increasing under post-processing and pre-processing:
Consider channel $\mathcal{E}=\mathcal{E}_2\circ \mathcal{E}_1$. Choose optimal Kraus operators $\{K_i\}$, $\{L_j\}$ for channels $\mathcal{E}_1,\mathcal{E}_2$ respectively. Explicit Kraus operators for $\mathcal{E}$ are then $\{L_{j}K_{i}\}$.
\begin{align*}
D_M(\mathcal{E})&\leq \max_{\rho} \sum_{i,j} \mathrm{Tr}[K_{i}^\dagger L_{j}^\dagger L_{j} K_i\rho]\log \rank(L_{j}K_{i}),\\
&\leq  \max_{\rho} \sum_{i,j} \mathrm{Tr}[K_{i}^\dagger L_{j}^\dagger L_{j} K_i\rho]\log \rank(K_{i}),\\
&= \max_{\rho} \sum_{i} \mathrm{Tr}[K_{i}^\dagger \left(\sum_{j} L_{j}^\dagger L_{j}\right) K_i\rho]\log \rank(K_{i})\\
&= \max_{\rho} \sum_{i} \mathrm{Tr}[K_{i}^\dagger K_i\rho]\log \rank(K_{i}) = D_M(\mathcal{E}_1).
\end{align*}
The above calculation proves our claim for post-processing. Following similar steps, but eliminating $K_i$ from the rank term on the second line and noting that the maximisation on the third line becomes a maximisation over the subset of states that are given by the outputs of the channel $\mathcal{E}_2$, i.e. over the states $\sum_i K_i\rho K_i^\dagger$, shows the claim for pre-processing.
\item Entanglement breaking channels have dimension measure zero.\\
It is already known \cite{horodecki2003entanglement} that a finite dimensional channel is entanglement breaking iff it has a Kraus representation with operators of rank 1. Thus $\log \rank K_i = 0,\; \forall i$ implying $D_M(\mathcal{E})=0$.\\
\item Isometric channels have dimension $\log d_1$, where $d_1$ is the dimension of the input space.\\
Isometric channels take the form $\mathcal{E}(\rho) = V\rho V^\dagger$, where $V$ is a $d_2\otimes d_1$ isometric operator. Thus, there is a single Kraus operator, whose rank is the smallest of its two dimensions. However, since $V^\dagger V = \openone_{d_1}$, its rank must be at least $d_1$.\\
\end{itemize}
\begin{theorem}
$D_M(\mathcal{E}) \geq \sup_{\ket{\phi}} S_M[id\otimes \mathcal{E}(\ket{\phi}\bra{\phi})]$.
\end{theorem}
To prove this result, we will use that $\rank(M)=\rank(M^\dagger M)$ - to see this, write $M$ in its singular values decomposition. Thus we can equivalently write $D_M(\mathcal{E})=\min_{K_i} \max_{\rho} \mathrm{Tr}(K_i^\dagger K_i \rho)\log\rank(K_i^{\dagger}K_i)$.\\

Next, we want to again rewrite an equivalent value for $D_M$. we claim that:
\begin{equation}\label{eq:twodefinitions}
D_M(\mathcal{E}) = \min_{K_i} \sup_{\rho} \mathrm{Tr}(K_i^\dagger K_i \rho)\log\rank(K_i^{\dagger}K_i\rho).
\end{equation}
$\geq:$ We note that, for a fixed set of $K_i$, $\sup_{\rho} \mathrm{Tr}(K_i^\dagger K_i \rho)\log\rank(K_i^{\dagger}K_i\rho) \leq \sup_{\rho} \mathrm{Tr}(K_i^\dagger K_i \rho)\log\rank(K_i^{\dagger}K_i)$. However, this is just a linear function of $\rho$, implying $\sup_{\rho} \mathrm{Tr}(K_i^\dagger K_i \rho)\log\rank(K_i^{\dagger}K_i)=\max_{\rho} \mathrm{Tr}(K_i^\dagger K_i \rho)\log\rank(K_i^{\dagger}K_i)$. Thus one way is proven.\\
$\leq:$ Again, let us fix $K_i$, and take $\hat{\rho}$ as the maximum achieving density matrix. For these two quantities not to be equal, we must have for at least one $i$ such that $\rank(K_i^\dagger K_i \hat{\rho})<\rank(K_i^\dagger K_i)$. We will denote $\mathcal{V}_i$ as the subspace of the support of $K_i^\dagger K_i$ orthogonal to the support of $\hat{\rho}$. Let us define a new state, 
\begin{equation}\label{SupDef}
\hat{\rho}^{\epsilon}:=(1-\epsilon)\hat{\rho} + \frac{\epsilon}{|\mathcal{I}|}\sum_i \openone_{\mathcal{V}^{i}}/\mathrm{Tr}(\openone_{\mathcal{V}^{i}}).
\end{equation}
This state is well defined, as we know the optimal set of Kraus operators is finite.\\
Note that, for any $\epsilon >0$, $\mathrm{rank}(K_i^\dagger K_i \hat{\rho}^{\epsilon})=\mathrm{rank}(K_i^\dagger K_i)$. In particular, this means that:
\begin{align*}
\sup_{\rho} \mathrm{Tr}(K_i^\dagger K_i \rho)\log\rank(K_i^{\dagger}K_i\rho)
&\geq \lim_{\epsilon\rightarrow 0} \mathrm{Tr}(K_i^\dagger K_i \hat{\rho}^\epsilon)\log\rank(K_i^{\dagger}K_i\hat{\rho}^\epsilon) \\
&= \lim_{\epsilon\rightarrow 0} \mathrm{Tr}(K_i^\dagger K_i \hat{\rho}^\epsilon)\log\rank(K_i^{\dagger}K_i)\\ &=\mathrm{Tr}(K_i^\dagger K_i \hat{\rho})\log\rank(K_i^{\dagger}K_i) 
\\&= D_M(\mathcal{E}).
\end{align*}
Thus this equivalence of definition in Equation (\ref{eq:twodefinitions}) is shown.\\

We now use a result \cite{WNotes} which tells us any pure state can be expressed as $\ket{\phi}=\sqrt{d\,\rho_{A}}V\otimes \openone\ket{\Phi^+_d}$, where $\rho_{A}$ is the reduced density matrix of $\ket{\phi}$ - we will associate its counterpart $\rho_{B}$ with the maximised state $\rho$. $V$ is a unitary to convert between the chosen basis of $\ket{\Phi^+_d}$, and the required Schmidt decomposition bases. Using this, we will show that a choice of Kraus operators corresponds to a decomposition choice of $id\otimes\mathcal{E}(\ket{\phi}\bra{\phi})$. We can do this via explicit calculation:
\begin{align*}
id\otimes\mathcal{E}(\ket{\phi}\bra{\phi})&=id\otimes\mathcal{E}(\sqrt{d\,\rho_{A}}V\otimes \openone\ket{\Phi^+_d}\bra{\Phi_d^+}V^{\dagger}\sqrt{d\rho_{A}}\,\otimes \openone),\\
&=\left(\sqrt{d\,\rho_A} \otimes \openone\right)\left(\sum_{k,a,b,i,j,r,s} \openone\otimes K^{k}_{ab}\ket{f_a}\bra{f_b} \left(\ket{e_i,f_i}\bra{e_j,f_j}/d\right)K^{k \dagger}_{rs}\ket{f_r}\bra{f_s}\right)\left(\sqrt{d\,\rho_A} \otimes \openone\right),\\
&=\left(\sqrt{\rho_A} \otimes \openone\right)\left(\sum_{k,i,a,j,s} K^{k}_{ai}\ket{e_i,f_a}\bra{e_j,f_s}K^{k \dagger}_{js}\right)\left(\sqrt{\rho_A} \otimes \openone\right),\\
&= \sum_{k}\left(\sum_{i,a}\lambda_{i}K^{k}_{ai}\ket{e_i,f_a}\right)\left(\sum_{j,s}\bra{e_j,f_s}\lambda_{j}K^{k \dagger}_{js}\right).
\end{align*}
We see that our choice of Kraus operators has lead to a specific decomposition of our output state into \emph{subnormalised} pure states.\\
If $\rho=\sum_{k}\ket{\phi^{k}}\bra{\phi^{k}}$ is a decomposition into subnormalised pure states, then we may write:
\begin{equation}
E_{F}(\rho)=\inf\sum_{k}\braket{\phi^{k}}{\phi^{k}}S_{R}(\ket{\phi^k}).
\end{equation}
For the explicit decomposition obtained above, we find the right-hand side is equal to:\\
\begin{align*}
&\sum_{k} (\sum_{i,a} K^{k \dagger}_{ia}K^{k}_{ai}\lambda_i^2)\,\mathrm{rank}\left(K^{k  \dagger}_{js}K^{k}_{si}\lambda_{i}\lambda_{j}\ket{e_i}\bra{e_j}\right)\\
\equiv&\sum_{k} (\sum_{i,a} K^{k \dagger}_{ia}K^{k}_{ai}\lambda_i^2)\,\mathrm{rank}\left(K^{k  \dagger}_{js}K^{k}_{si}\lambda_{i}\lambda_{j}\ket{f_i}\bra{f_j}\right)\\
=&\sum_{k} \mathrm{Tr}(K^{k \dagger} K^{k} \rho_B)\,\mathrm{rank}\left(\sqrt{\rho_{B}}K^{k \dagger}K^{k}\sqrt{\rho_{B}}\right)\\
=&\sum_{k} \mathrm{Tr}(K^{k \dagger} K^{k} \rho_B)\,\mathrm{rank}\left(K^{k \dagger}K^{k}\rho_{B}\right)
\end{align*}
Where in the last step we have used that for products of square matrices, the eigenvalues (and thus the rank) are invariant under cyclic permutations \cite{HJ2013}.\\

Let us now choose a specific decomposition of $id\otimes\mathcal{E}(\ket{\phi}\bra{\phi})$. We can write this as:
\begin{equation}
id\otimes\mathcal{E}(\ket{\phi}\bra{\phi})=\sum_{k} \left(\sum_{i,j}\alpha^{k}_{i,j}\ket{e_i,f_j}\right)\left(\sum_{r,s}\bra{e_r,f_s}\alpha^{k *}_{r,s}\right)
\end{equation}
where we again use subnormalised states for convenience. Note that we require the partial trace of this state to be $\rho_{A}$, thus implying: $\sum_{k} \sum_{j}\alpha^{k}_{i,j}\alpha^{k *}_{r,j} = \delta_{i,r}\lambda_{i}^2$. Note that if $\lambda_{i}=0$, this implies $\bs{\alpha}=0$.\\

Equating the two decompositions, we see that so long as the Kraus operators defined by $K^{k}_{ji}=\alpha^{k}_{i,j}/\lambda_{i}$ ($\lambda_i\neq 0$) are valid, then we can reach any pure state decomposition. Checking, we see that
\begin{align*}
\sum_{k,j} K^{k \dagger }_{rj}K^{k}_{ji} = \sum_{k,j} \alpha^{k}_{i,j}\alpha^{*}_{r,j}/\lambda_{i}\lambda_{r} = \lambda_{i}^2\delta_{ir}/\lambda_i\lambda_r = 1.
\end{align*}
This proves that, on the support of $\rho_B$, the Kraus operators are the identity. To make this a proper channel, we append to these the operators $\{\ket{f_i}\bra{f_i}\}$ which span the kernel of $\rho_{B}$ (when $\lambda_i=0$). This gives us a proper channel, whilst not altering the quantity $\sum_{k} \mathrm{Tr}(K^{k \dagger} K^k \rho_B)\,\mathrm{rank}\left(K^{k \dagger}K^{k}\rho_{B}\right)$.\\

This equivalence allows us to write:
\begin{equation}
\sup_{\phi} S_M[id\otimes \mathcal{E}(\ket{\phi}\bra{\phi})] =  \sup_{\rho} \min_{K_i} \mathrm{Tr}(K_i^\dagger K_i \rho)\log\rank(K_i^{\dagger}K_i\rho).
\end{equation}
Swapping the order of the minimum and supremum can only increase the right hand side (this is the max–min inequality); but doing so gives us exactly $D_{M}(\mathcal{E})$. Thus we have proven our inequality.

\section{Properties of the Measurement Dimension Measure}\label{PropMeasDim}

\begin{itemize}
\item Convexity.\\
For two sets of measurements $\bs{M}=\{M_{a|x}\}, \bs{N}=\{N_{a|x}\}$, let us define $\bs{Q}=\{Q_{a|x}=pM_{a|x}+(1-p)N_{a|x}\}$. If the optimal channel/measurements of $\bs{M}$ are given by $\{K_{i},R_{a|x}\}$ and of $\bs{N}$ by $\{L_{j},S_{a|x}\}$, then we have an explicit realisation of $\bs{Q}$ via $\{\sqrt{p}K_{i}\otimes \ket{0},\sqrt{1-p}L_j\otimes \ket{1}\}$ and $\{R_{a|x}\otimes \ket{0}\bra{0},S_{a|x}\otimes \ket{1}\bra{1}\}$.This tells us that:
\begin{align*}
D_M(\bs{Q}) &\leq \max_{\rho} \sum_{i} \mathrm{Tr}[pK_{i}^{\dagger}K_i \rho]\log \rank (K_{i}\otimes \ket{0}) + \sum_{j} \mathrm{Tr}[(1-p)L_{j}^{\dagger}L_j \rho]\log \rank (K_{j}\otimes \ket{1})\\
&\leq p\max_{\rho_1} \sum_{i} \mathrm{Tr}[K_{i}^{\dagger}K_i \rho_1]\log \rank (K_{i}) + (1-p)\max_{\rho_2} \sum_{j} \mathrm{Tr}[L_{j}^{\dagger}L_j \rho_2]\log \rank (L_{j}) \\
&= p D_M(\bs{M})+ (1-p) D_M(\bs{N}).
\end{align*}
\item It is subadditive over tensor product.\\
Defining $Q_{a_1,a_2|x_1,x_2}=M_{a_1|x_1}\otimes M_{a_2|x_2}$, we again take $\{K_{i},R_{a|x}\}$ optimal for $\bs{M}$, $\{L_{j},S_{a|x}\}$ optimal for $\bs{N}$. Then we have an explicit realisation of $\{Q_{a_1,a_2|x_1,x_2}\}$ via $\{K_{i}\otimes L_{j}\}$, $\{R_{a_1|x_1}\otimes S_{a_2|x_2}\}$. Subadditivity directly follows from our proof of subadditivity for channels above.\\
\item It is non-increasing under classical post-processing by linearity of channels.
For the quantum pre-processing, we take a pre-processing channel $\mathcal{E}_2$ and define $Q_{a|x}=\mathcal{E}_2^*(M_{a|x})$. Now, assuming that $M_{a|x}=\mathcal{E}_1^*(N_{a|x})$, then $Q_{a|x}=\mathcal{E}_2^*\circ\mathcal{E}_1^*(N_{a|x})$. The proof then follows directly from the result for channels above.\\
\item Jointly measurable measurements have compression measure 0.\\
For this proof, it is more convenient to use the second formulism via the pseudo-measurements $\bs{A}^{i}$. We know a set of measurements is jointly measureable iff they have a decomposition
$M_{a|x}=\sum_{\lambda(x)=a} G^{\lambda}$, where $G^{\lambda}$ are positive operators. Define $G^{k,\lambda}$ as the $k$\textsuperscript{th} eigenprojector of $G^{\lambda}$ multiplied by the corresponding eigenvalue. Then we can decompose
\begin{equation}
M_{a|x} = \sum_{\lambda,k} G^{k,\lambda}_{a|x}
\end{equation}
where $G^{k,\lambda}_{a|x}:=\delta_{a,\lambda(x)}G^{k,\lambda}$. Moreover, $\sum_{a} G^{k,\lambda}_{a|x} = G^{k,\lambda}$ which is rank 1, implying $\log \rank G^{k,\lambda} =0$ for all elements of this decomposition. This proves the result. 
\end{itemize}

\begin{lemma}
If the set of unitaries $\{U_j\}_{j=1}^N$ induces a relabelling of the measurement operators $U^{\dagger}_{j}M_{a|x}U_{j} = M_{\pi_{j,x}(a)|\pi_j(x)}$, then the maximum for the dimension measure is achieved a state of the form $\rho=1/d\sum_j U_{j}\rho U_{j}^{\dagger}$.
\end{lemma}
Suppose we have a set of measurements $\bs{N}$ and a channel $\mathcal{E}$, whose Kraus operators are $K_i$, such that $\mathcal{E}^*(N_{a|x})=M_{a|x}$. We define a new set of measurements and a new set of measurements via the Kraus operators $\{1/\sqrt{N} K_iU_j\otimes \ket{j}\}$, and measurements $N'_{a|x}:=\sum_{j} N_{\pi^{-1}_{j,x}(a)|\pi^{-1}_j(x)}\otimes \ket{j}\bra{j}$. We see that 
\begin{equation}
\sum_{i,j} \left(1/\sqrt{N}U_j^{\dagger}K_i^\dagger\otimes \bra{j}\right) \left(1/\sqrt{N}K_i U_j \otimes \ket{j}\right) = 1/N\sum_{j} U_{j}^{\dagger}\left(\sum_{i}K_i^{\dagger} K_i\right) U_{j} = 1/N \sum_{j} \mathbb{I} = \mathbb{I},
\end{equation}
and 
\begin{equation}
\sum_{a} N'_{a|x} = \sum_{a} \sum_{j} N_{\pi^{-1}_{j,x}(a)|\pi^{-1}_j(x)}\otimes \ket{j}\bra{j} = \sum_{j} \sum_{a'} N_{a'|\pi^{-1}_j(x)} \otimes \ket{j}\bra{j} = \mathbb{I} \otimes \sum_{j} \ket{j}\bra{j},
\end{equation}
which is the identity on our new, larger space. This means we have chosen a valid state and measurement. Furthermore,
\begin{align*}
\sum_{j,i} \left(1/\sqrt{N}U_j^{\dagger}K_i^\dagger\otimes \bra{j}\right) N'_{a|x} \left(1/\sqrt{N}K_i U_j \otimes \ket{j}\right) &= 
1/N \sum_{j,i}  U_j^{\dagger}K_i^\dagger \otimes \bra{j} \left(\sum_{j'} N_{\pi^{-1}_{j',x}(a)|\pi^{-1}_{j'}(x)}\otimes \ket{j'}\bra{j'}\right) K_i U_j \otimes \ket{j}\\
&= 
1/N \sum_{j,i}  U_j^{\dagger}\left(\sum_i K_i^\dagger N_{\pi^{-1}_{j',x}(a)|\pi^{-1}_{j'}(x)} K_i \right) U_j\\
&= 1/N \sum_{j} U^{\dagger}_j M_{\pi^{-1}_{j',x}(a)|\pi^{-1}_{j'}(x)}U_j\\
&= 1/N \sum_{j} M_{a|x} \\
&= M_{a|x}.
\end{align*}
This means, given any candidate channel and measurement set in our optimisation, we can construct a new channel and measurements with the above form. Let us look at how the dimension measure is changed with this new construction. We have that $\mathrm{rank}\left(U_jK_{i}\otimes \ket{j}\right)= \mathrm{rank}\, K_i$, and $1/N\sum_{j}\mathrm{Tr}\left[U_{j}^\dagger K_i^{\dagger}K_i U_{j} \rho\right] = \mathrm{Tr}\left[K_i^{\dagger}K_i U_{j} \left(1/N\sum_{j}U_{j}\rho U_j^{\dagger}\right)\right]$. As a consequence of this, we see that if we replace every candidate channel and measurement via this construction, we obtain almost the exact same optimisation: except that the state is now constrained to be of the form $1/N\sum_{j}U_{j}\rho U_j^{\dagger}$ - since this restricts our maximum, it can only improve our optimal value. Thus, the optimal is achieved by such a state; and the result is shown.\\
\begin{lemma}
If the set $\{X^iZ^j\}_{i,j=0}^{d-1}$, where $Z\ket{k}=\exp{\frac{2\pi i k}{d}}\ket{k}$ and $X\ket{k}=\ket{k+1}$ induce a relabelling of the measurement operators $U^{\dagger}_{j}M_{a|x}U_{j} = M_{\pi_{j,x}(a)|\pi_j(x)}$, then the maximum for the dimension measure is achieved by the state $\rho=\mathbb{I}/d$.
\end{lemma}
We have already seen the optimal will be achieved by a state of the form $1/d^2\sum_{i,j}Z^{-j}X^{-i}\rho X^iZ^j$. Let us express $\rho$ in the basis $1/d\sum_{k,l} c_{k,l}X^kZ^l$, where $c_{0,0}=1$. One can verify that $Z^{-j}X^{-i}X^kZ^lX^iZ^j =\omega^{il-jk}X^kZ^l$, where $\omega=e^{2\pi i/d}$. Thus we have that 
\begin{equation}
1/d^2\sum_{i,j}Z^{-j}X^{-i}\rho X^iZ^j = 1/d^{3}\sum_{k,l,i,j} c_{k,l}\omega^{il-jk}X^kZ^l = 1/d^{3}\sum_{k,l} c_{k,l}\sum_{i}\omega^{il}\sum_{j}\omega^{-jk}X^kZ^l.
\end{equation}
However, $\sum_{i}\omega^{il}=d\delta_{l0}$ and $\sum_{j}\omega^{-jk}=d\delta_{k0}$, meaning all terms in this sum vanish except for $c_{0,0}=1$. Thus we may conclude that $1/d^2\sum_{i,j}Z^{-j}X^{-i}\rho X^iZ^j=\mathbb{I}/d$.\\

\section{Properties of the Steering Entanglement of Formation}\label{SEFA}
In this section we prove the continuity of the steering entanglement of formation on the set of extremal assemblages. Note that, in general, it is not continuous; see \cite{C2021} for a counterexample.
\begin{lemma}
The steering entanglement of formation for assemblages ($E_{\mrm{FA}}$) is continuous on extremal assemblages.\\
\end{lemma}
\begin{proof}
On extremal assemblages, $E_{\mrm{FA}}$ coincides with $S(\rho_{\bs{\sigma}})$. Suppose we have two assemblages which are $\epsilon$ close. We are using the norm:
\begin{equation}
\|\bs{\sigma}-\bs{\tau}\| = \frac{1}{|\mathcal{X}|}\sum_{a,x} \|\sigma_{a|x}-\tau_{a|x}\|_{1},
\end{equation}
where $\|\cdot\|_1$ is the trace distance.\\
We can now use the chain of inequalities:
\begin{align*}
\epsilon \geq \|\bs{\sigma}-\bs{\tau}\| &= \frac{1}{|\mathcal{X}|}\sum_{a,x} \|\sigma_{a|x}-\tau_{a|x}\|_{1}\\
&=\frac{1}{|\mathcal{X}|}\sum_{x} \left(\sum_{a} \|\sigma_{a|x}-\tau_{a|x}\|_{1}\right) \\
&\geq \frac{1}{|\mathcal{X}|}\sum_{x} \left(\| \sum_{a} \sigma_{a|x}-\tau_{a|x}\|_{1}\right)\\
&= \frac{1}{|\mathcal{X}|}\sum_{x} \|\rho_{\bs{\sigma}}-\rho_{\bs{\tau}}\|_1 \\
&=\|\rho_{\bs{\sigma}}-\rho_{\bs{\tau}}\|_1
\end{align*}
We can now apply Fannes inequality: assuming that Bob's subsystem is of dimension $d$,
$\|\rho_{\bs{\sigma}}-\rho_{\bs{\tau}}\|_1 \leq \epsilon$ implies $|S(\rho_{\bs{\sigma}})-S(\rho_{\bs{\tau}})|\leq (\epsilon\log(d-1))+ H_{2}(\epsilon)$.\\ As the right hand side tends to $0$ as $\epsilon\rightarrow 0$, we have shown continuity.\\
\end{proof}

\section{Additional Calculations for Asymptotics}
In this appendix we provide additional detail for some of the calculations presented in the main paper. In particular, we show the step-by-step calculations for calculating the relevant trace-norms.
\begin{lemma}
$\|(\rho^{\otimes n})^{\otimes m}- \rho^{n,m}\|_1 \leq \mathbb{E}[P(i\not \in T_{\delta}(I)]_{\mathbf{k}}$.
\end{lemma}
We remind the reader that $(\rho^{\otimes n})^{\otimes m}$ is decomposed into the sum $\sum_{\mathbf{k}} p_\mathbf{k}(\ket{\Phi_{\mathbf{k}}}\bra{\Phi_{\mathbf{k}}}$, which is $m$ copies of the decomposition achieving $E_F(\rho^{\otimes n})$. We then define $\rho^{n,m}$ via $\rho^{n,m}=\sum_{\mathbf{k}} p_\mathbf{k} \ket{\Phi'_{\mathbf{k}}}\bra{\Phi'_{\mathbf{k}}}$, where $\ket{\Phi'_{\mathbf{k}}}=\sum_{\mathbf{i}\in T_{\delta}(I|k)} \boldsymbol{\lambda}_{\bs{i}|\bs{k}}\ket{\boldsymbol{e_{i|k}f_{i|k}}}$, and $\boldsymbol{\lambda}_{\bs{i}|\bs{k}}=\lambda_{i_1|k_1}\lambda_{i_2|k_2}\ldots,\; \boldsymbol{e_{i|k}}=e_{i_1|k_1}\otimes e_{i_2|k_2}\ldots$. 
These definitions allows us to bound the trace norm distance:

\begin{align*}
\|(\rho^{\otimes n})^{\otimes m}- \rho^{n,m}\|_1&=\|\sum_{\mathbf{k}} p_\mathbf{k}(\ket{\Phi_{\mathbf{k}}}\bra{\Phi_{\mathbf{k}}}-\ket{\Phi'_{\mathbf{k}}}\bra{\Phi'_{\mathbf{k}}})\|_1\\
&\leq \sum_{\mathbf{k}} p_\mathbf{k} \|\ket{\Phi_{\mathbf{k}}}\bra{\Phi_{\mathbf{k}}}-\ket{\Phi'_{\mathbf{k}}}\bra{\Phi'_{\mathbf{k}}})\|_1\\
& = \sum_{\mathbf{k}} p_\mathbf{k}\sum_{i\not\in T_{\delta}(I)} \lambda^2_{i|k} =\mathbb{E}[P(i\not \in T_{\delta}(I)]_{\mathbf{k}}.
\end{align*}

In the case of assemblages, we decomposed $(\bs{\sigma}^{\otimes n})^{\otimes m} = \sum_{\mathbf{k}} p_{\bf{k}}\bs{\sigma}_{\bf{k}}$, using $m$ copies of the optimal decomposition for the steering entanglement of formation $E_{FA}(\bs{\sigma}^{\otimes n})$. For each $\bs{\sigma}_{\bf{k}}$, we used the purification of $\rho_{\bs{\sigma}_{\bf{k}}}$,  $\ket{\Phi_{\bf{k}}}=\sum_{\mathbf{i}}\lambda_{\bf{i}|\bf{k}}\ket{e_{\mathbf{i}|\bf{k}}}\ket{e_{\mathbf{i}|\bf{k}}}$ and measurements $M_{\bf{a}|\bf{x}}=\rho^{-1/2}_{\bs{\sigma}_{\bs{k}}}\sigma^{T}_{\bs{a}|\bf{x}}\rho^{-1/2}_{\bs{\sigma}_{\bs{k}}}$, to obtain a construction of this assemblage. For each $\bs{k}$ we created a new (subnormalised) assemblage by replacing the state $\ket{\Phi_{\bf{k}}}$ with the state $\ket{\Phi'_{\mathbf{k}}}$, where $\ket{\Phi'_{\mathbf{k}}}=\sum_{\mathbf{i}\in T_{\delta}(I|k)} \boldsymbol{\lambda}_{\bs{i}|\bs{k}}\ket{\boldsymbol{e_{i|k}f_{i|k}}}$. This new assemblage we labelled as $\bs{\sigma}'_{\bs{k}}$.

First, we calculate the distance between $\bs{\sigma}'_{\bs{k}}$ and $\bs{\sigma}_{\bs{k}}$.

\begin{align*}
\|\bs{\sigma}_{\bs{k}}-\bs{\sigma}'_{\bs{k}}\|_1=&\sum_{\bs{a},\bs{x}} p(\bs{x}) \|\sigma_{\bs{k},\bs{a}|\bs{x}} - \sigma'_{\bs{k},\bs{a}|\bs{x}}\|_1 \\
= &\sum_{\bs{a},\bs{x}} p(\bs{x}) \|\text{Tr}_{A}\left[\left(M_{\bs{a}|\bs{x}}\otimes \openone\right)\left(\ket{\Phi_{\bf{k}}}\bra{\Phi_{\bf{k}}}-\ket{\Phi'_{\mathbf{k}}}\bra{\Phi'_{\mathbf{k}}}\right)\right]\|_1\\
\leq &\sum_{\bs{a},\bs{x}} p(\bs{x}) \|\left(M_{\bs{a}|\bs{x}}\otimes \openone\right)\left(\ket{\Phi_{\bf{k}}}\bra{\Phi_{\bf{k}}}-\ket{\Phi'_{\mathbf{k}}}\bra{\Phi'_{\mathbf{k}}}\right)\|_1\\
= & \sum_{\bs{a},\bs{x}} p(\bs{x}) \|\left(M_{\bs{a}|\bs{x}}\otimes \openone\right)\left(\ket{\Phi_{\bf{k}}}\bra{\Phi_{\bf{k}}}-\ket{\Phi'_{\mathbf{k}}}\bra{\Phi'_{\mathbf{k}}}\right)\otimes \ket{\bs{a}}\bra{\bs{a}}\|_1\\
= & \sum_{\bs{x}} p(\bs{x}) \|\sum_{\bs{a}} \left(M_{\bs{a}|\bs{x}}\otimes \openone\right)\left(\ket{\Phi_{\bf{k}}}\bra{\Phi_{\bf{k}}}-\ket{\Phi'_{\mathbf{k}}}\bra{\Phi'_{\mathbf{k}}}\right)\otimes \ket{\bs{a}}\bra{\bs{a}}\|_{1}\\
\leq& \sum_{\bs{x}} p(\bs{x}) \|\left(\ket{\Phi_{\bf{k}}}\bra{\Phi_{\bf{k}}}-\ket{\Phi'_{\mathbf{k}}}\bra{\Phi'_{\mathbf{k}}}\right)\|_{1} \\
=& \|\left(\ket{\Phi_{\bf{k}}}\bra{\Phi_{\bf{k}}}-\ket{\Phi'_{\mathbf{k}}}\bra{\Phi'_{\mathbf{k}}}\right)\|_{1}
\end{align*}
Where we have used the fact that the 1 norm is additive over orthogonal subspaces, and that it is non-increasing under the action of a CPTP channel (the above is an application of a quantum instrument).\\

We then construct our approximation to $(\bs{\sigma}^{\otimes n})^{\otimes m}$ , which we call $\bs{\sigma}^{n,m}$, via $\bs{\sigma}^{n,m}= \sum_{\bs{k}} p_{\bs{k}}\sigma'_{\bs{k}}$. We can bound the distance between these two assemblages in the following way:
\begin{align*}
\|\left(\bs{\sigma}^{\otimes n}\right)^{\otimes m} - \bs{\sigma}^{n,m}\|_1&=\|\sum_{\bs{k}} p_{\bs{k}}(\bs{\sigma}_{\bs{k}}-\bs{\sigma}'_{\bs{k}})\|_1\\
&\leq \sum_{\bs{k}} p_{\bs{k}} \|\bs{\sigma}_{\bs{k}}-\bs{\sigma}'_{\bs{k}}\|_1\\
&\leq \sum_{\bs{k}} p_{\bs{k}} \|\ket{\Phi_{\bf{k}}}\bra{\Phi_{\bf{k}}} - \ket{\Phi'_{\bf{k}}}\bra{\Phi'_{\bf{k}}}\|_1\\
& =\mathbb{E}[P(i\not \in T_{\delta}(I)]_{\mathbf{k}}.
\end{align*}

\section{Symmetries of the canonical MUB pair}\label{MUBSym}
We have already seen how MUBs constructed via the operators $Z\ket{k}=\exp{\frac{2\pi i k}{d}}\ket{k}$ and $X\ket{k}=\ket{k+1}$ satisfy enough symmetry to restrict the maximisation term in the dimension measure to a single state. If we consider only the pair of MUBs given by the eigenvectors of the operators $Z,X$, however, we can prove a much stronger result: there exists a set of ``unitary relabelling equivalences'', as introduced in the main body of the paper, which uniquely define measurements of the form:
\begin{equation}\label{GPauliAppendix}
M^{p}_{a|x} = p\ket{e^x_a}\bra{e^x_a}+(1-p)\openone/d 
\end{equation}
where $x=0$ labels the eigenvectors of $X$ and $x=1$ of $Z$.\\
By averaging, or ``twirling" over these unitary-relabellings, one can transform an arbitrary $|X|=2,|A|=d$ pseudo-measurement into one proportional to the form (\ref{GPauliAppendix}).\\

Before we begin the proof, it will be useful to know the unitaries under which $\bs{M}^{p}$ is invariant. these are: $\{X^{i}Z^{j}\}$, $\{P_{k}\}$, the permutation matrices which perform an automorphism of the $Z$ eigenbasis, and the Fourier transform. The full set of unitaries we will consider is $\{U_j^\dagger\}=\{F^{-f}P_{\alpha}X^{i_{1}}Z^{i_{0}}\}_{f=0,\alpha=1,i_1=0,i_0=0}^{f=1,\alpha=d-1,i_1=d-1,i_0=d-1}$.\\

In this proof, we will take all output labels to be modulo $d$, and input labels modulo 2. We take an arbitrary $|X|=2,|A|=d$  pseudo-measurement $\bs{A}$. Then we can define $d^2$ pseudo-measurements of the same measure via $A^{i_1,i_2}_{a|x}:=X^{i_1}Z^{i_0}A_{a-i_x|x}(X^{i_1}Z^{i_0})^\dagger$.

Let us now consider $\tilde{A}=1/d^2\sum_{i_1,i_2}A^{i_1,i_2}=1/d^2\sum_{i_1,i_2}X^{i_1}Z^{i_0}A(X^{i_1}Z^{i_0})^{\dagger} = \mathrm{Tr}[A]\openone_d$. This means that $\tilde{A}_{a|0}$ is (proportional to) a valid set of measurements such that $Z^{i_0}\tilde{A}_{a|0}(Z^{i_0})^{\dagger} = \tilde{A}_{a+i_0|0}$ and $X^{i_1}\tilde{A}_{a|1}(X^{i_1})^{\dagger} = \tilde{A}_{a+i_1|1}$. From \cite{CHT2012} we know that means the operators must take the form:
\begin{align*}
\tilde{A}_{a|0} &= \sum_{b} r_{b-a}\ket{x_b}\bra{x_b}, & \tilde{A}_{a|1} &= \sum_{b} s_{b-a}\ket{z_b}\bra{z_b}.
\end{align*}

Next, let us consider an automorphism of $(\mathbb{Z}/d\mathbb{Z})$; there are of the form: $z\rightarrow \alpha z$, $z \in (\mathbb{Z}/d\mathbb{Z}), \alpha \in (\mathbb{Z}/d\mathbb{Z})^*\equiv \{1\ldots d-1\}$. It is a known result \cite{HZ2019} that the permutation matrix inducing $\alpha$ in the $z$ basis induces the automorphism $\alpha^{-1}$ in the $x$ basis, where $\alpha^{-1}$ is the multiplicative inverse of $\alpha$ mod $d$. We can verify that $M^{p}_{a|x}=1/(d-1)\sum_{\alpha}P_{\alpha}M^{p}_{\alpha^{-x}a|x}P^{\dagger}_{\alpha}$. In the same vein as before, we can now define a new series of pseudo-measurements via $A^{i_1,i_2,\alpha}_{a|x}:=P_{\alpha}X^{i_1}Z^{i_0}A_{\alpha^{-x}(a-i_x)|x}(X^{i_1}Z^{i_0})^\dagger P_{\alpha}^{\dagger}$. We can now consider the average over all of these $d^3-d$ pseudo-measures, which by linearity is $\hat{A}_{a|x}=1/(d-1)\sum_{\alpha}P_{\alpha}\tilde{A}_{\alpha^{-x}a|x}P_{\alpha}^{-1}$. Let us have a look at this sum when x=1:
\begin{align*}
	\hat{A}_{a|1}&=1/(d-1)\sum_{\alpha}P_{\alpha}\tilde{A}_{\alpha^{-1}a|1}P_{\alpha}^{-1}\\
	&= 1/(d-1)\sum_{\alpha}P_{\alpha}\sum_{b}s_{b-\alpha^{-1}a}\ket{z_b}\bra{z_b}P_{\alpha}^{-1}\\
	&= 1/(d-1)\sum_{\alpha}\sum_{b}s_{b-\alpha^{-1}a}\ket{z_{\alpha b}}\bra{z_{\alpha b}}\\
	&= 1/(d-1)\sum_{\alpha}\sum_{b'}r_{b'}\ket{z_{\alpha b' +\alpha \alpha^{-1} a}}\bra{z_{\alpha b' +\alpha \alpha^{-1} a}}\\
	&= 1/(d-1)\sum_{\alpha}\sum_{b'}s_{b'}\ket{z_{\alpha b' +a}}\bra{z_{\alpha b' +a}}\\
	&= s_0\ket{z_a}\bra{z_a} +\sum_{\beta\neq 0}\sum_{b'\neq 0}s_{b'}/(d-1)\ket{z_{\beta+a}}\bra{z_{\beta+a}}
\end{align*}
where we have used that $\alpha$ acts transitively on $(\mathbb{Z}/d\mathbb{Z})^*$ and leaves $0$ invariant.\\
From this we see the symmetry of $\bs{M}^{p}$ forces $\hat{A}_{a|1}$ to have a single unique eigenvalue associated with eigenvector $\ket{z_a}$, with all other eigenvalues coinciding with each other. The same proof can be used for $x=0$.\\

Finally, we note that $\bs{M}^{p}$ has the symmetry $M^{p}_{a|x}=F^{-1}M^{p}_{(-1)^{x}a|x\oplus 1}F$, where $F$ is the Fourier transform. This allows us to further increase the number of newly defined pseudo-measurements: 
\begin{equation}
A^{i_1,i_2,\alpha,f}_{a|x}:=F^{-f}P_{\alpha}X^{i_1}Z^{i_0}A_{(-1)^{fx}\alpha^{-x}(a-i_x)|x\oplus f}(X^{i_1}Z^{i_0})^\dagger  P_{\alpha}^{\dagger}F^{f}.
\end{equation} Again by linearity we consider $\bar{A}_{a|x} = 1/2(\hat{A}_{a|x}+F^{-1}\hat{A}_{(-1)^{x}a|x\oplus 1}F)$. We see this is equal to: $1/2\sum_{b} r'_{b-a} + s'_{b-a}\ket{x_b}\bra{x_b}$, where $r'_0 = r_{0}$ and $r'_{\neq 0} = 1/\sum_{b\neq 0}r_b/(d-1)$, and similarly for $s'$. In particular, this means that $\bar{\bs{A}}$ is proportional to some $\bs{M}^{q}$, with the proportionality given by $\mathrm{Tr}[A]$.

Before we move on - let us comment on the importance of considering prime dimensions: the automorphism group acts transitively on $\{1\ldots d-1\}$, leading to all $r_i$ except $r_0$ being averaged out. In non-prime dimensions this doesn't happen; however it may be there exist another set of unitary-relabellings which achieve the same effect.\\

\subsection{Consequences for the Dimension Measure and Compression Dimension}
The ability to ``twirl" pseudo-measurements has important consequences when calculating the dimension measure for the noisy MUB pair given in Eq. (\ref{GPauliAppendix}). Suppose we have a decomposition of $\bs{M}^p=\sum_{k}\bs{A}^k$. For the set $\{U_j^\dagger\}$ described in the previous section, we have that:
\begin{align*}
{M}^p_{a|x}&=\frac{1}{2(d^3-d)}\sum^{2(d^3-d)}_{j=1}U_j^{\dagger}{M}^p_{\pi^{-1}_{j,x}(a)|\pi^{-1}_{j}(x)}U_j \\
&= \sum_{k}\frac{1}{2(d^3-d)}\sum_{j}^{2(d^3-d)}U_j^{\dagger}{A}^k_{\pi^{-1}_{j,x}(a)|\pi^{-1}_{j}(x)}U_j\\
&= \sum_{k}\frac{\mathrm{Tr}[A^k]}{d}{M}^{p_k}_{a|x}.
\end{align*}
Since unitaries will leave $\mathrm{rank}\,A^k$ invariant, it is clear that our new decomposition can only decrease the dimension measure - and therefore there exists an optimal decomposition of this form 
(we already know that $\rho$ can be taken to be $\mathbb{I}/d$). Therefore the problem of finding the dimension measure of $\bs{M}^p$ reduces to finding the range of $p$ for each $n$ which can be constructed from rank $n$ pseudo-measurements, i.e. have compression dimension $n$.\\

This question can be phrased in the following way: $\max_{p_n}$ such that $\bs{M}^{p_n} = \sum_{k} \bs{A}^k$, $\rank\,A^k \leq n$. However, by applying the same technique as above, we see that
\begin{align*}
{M}^{p_n}_{a|x}&= \sum_{k}\frac{1}{2(d^3-d)}\sum_{j}^{2(d^3-d)}U_j^{\dagger}{A}^k_{\pi^{-1}_{j,x}(a)|\pi^{-1}_{j}(x)}U_j\\
&= \sum_{k}\frac{\mathrm{Tr}[A^k]}{d}{M}^{p_k}_{a|x}.
\end{align*}
In order for the $p_n$ to be maximal, we therefore require that all $p_k$ co-incide with $p_n$ itself. Consequently, the maximum can be achieved by a decomposition into a single rank $n$ pseudo-measurement, and its corresponding unitary-relabellings.\\

Finally, we note that we can rewrite $p$ in the following way:
\begin{equation}
p = \frac{d\sum_{a,x}\mathrm{Tr}[M^{1}_{a|x}M^{p}_{a|x}]}{(d-1)|\mathcal{X}||\mathcal{A}|}-\frac{1}{d-1}
\end{equation}
let us substitute our decomposition of $\bs{M}^{p_n}$ into the trace term. We see that
\begin{align*}
\sum_{a,x}\mathrm{Tr}[M^{1}_{a|x}M^{p_n}_{a|x}] &= \sum_{a,x}\mathrm{Tr}\left[M^{1}_{a|x}\frac{1}{2(d^3-d)}\sum_{j}^{2(d^3-d)}U_j^{\dagger}{A}_{\pi^{-1}_{j,x}(a)|\pi^{-1}_{j}(x)}U_j\right]\\
&= \frac{1}{2(d^3-d)}\sum_{j}^{2(d^3-d)}\sum_{a,x}\mathrm{Tr}\left[U_j M^{1}_{a|x}U_j^{\dagger}{A}_{\pi^{-1}_{j,x}(a)|\pi^{-1}_{j}(x)}\right]\\
&=\frac{1}{2(d^3-d)}\sum_{j}^{2(d^3-d)}\sum_{a,x}\mathrm{Tr}\left[M^{1}_{\pi^{-1}_{j,x}(a)|\pi^{-1}_{j}(x)}{A}_{\pi^{-1}_{j,x}(a)|\pi^{-1}_{j}(x)}\right]\\
&=\sum_{a,x}\mathrm{Tr}[M^{1}_{a|x}{A}_{a|x}]\\
\end{align*}
where we have used the property that the permutation leaves the sum invariant.\\

This means $p_n$, the largest $p$ for which $\bs{M}^p$ has compression dimension $n$, is given by
\begin{align}
\max_{\bs{A}} \frac{d\sum_{a,x}\mathrm{Tr}[M^{1}_{a|x}A_{a|x}]}{(d-1)|\mathcal{X}||\mathcal{A}|}-\frac{1}{d-1}
\end{align}
such that $\bs{A}$ is a pseudo-measurement with $\rank\,A \leq n$, $\mathrm{Tr}[A]=d$. Due to the rank constraint we were unable to solve this problem; our best answer is given by the construction in the main body of the paper.

\section{Derivation of Eq. \ref{SR2}}\label{Eq2to3}
In this section, we explain how to derive Eq. \ref{SR2} from Eq. \ref{ChoiDecomp}, in the case where either the Choi matrix belongs to either $L(\mathcal{H}_2\otimes \mathcal{H}_n)$ or $L(\mathcal{H}_n\otimes \mathcal{H}_2)$ - which is the case when either the input or output space of the channel has dimension 2.

For this space, any pure state can have Schmidt rank either 1 or 2. In particular, this means we can split our decomposition of $\chi_{\mathcal{E}}$ into two disjoint sets:
\begin{equation}
\chi_{\mathcal{E}}=\sum_{l,S_R(\phi_l)=1} \ket{\phi_l}\bra{\phi_l} +\sum_{l,S_R(\phi_l)=2} \ket{\phi_l}\bra{\phi_l} = \sigma_0 + \sum_{l,S_R(\phi_l)=2}
\end{equation}
where we note that $\sigma_0$ is a (subnormalised) separable state.\\

We can plug this decomposition into Eq. \ref{ChoiDecomp}, where we find:
\begin{align*}
&\min_{\phi^l} \max_\rho \sum_l d\,\mathrm{Tr}(\rho^{T}\otimes \openone\ket{\phi^l}\bra{\phi^l})\log S_R(\ket{\phi^l})\\ =&\min_{\phi^l} \max_\rho \sum_{l,S_R(\phi_l)=1} d\,\mathrm{Tr}(\rho^{T}\otimes \openone\ket{\phi^l}\bra{\phi^l})\log 1 + \sum_{l,S_R(\phi_l)=2} d\,\mathrm{Tr}(\rho^{T}\otimes \openone\ket{\phi^l}\bra{\phi^l})\log 2\\
=&  d\,\mathrm{Tr}\left(\rho^{T}\otimes \openone \sum_{l,S_R(\phi_l)=2}\ket{\phi^l}\bra{\phi^l})\right)\\
=&  d\,\mathrm{Tr}\left(\rho^{T}\otimes \openone \left(\chi_{\mathcal{E}}-\sigma_0\right)\right)
\end{align*}
Since only $\sigma_0$ appears as an unknown in this equation, we change the decomposition over $\{\phi_k\}$ to one directly over $\sigma_0$ instead. This gives us the familiar form of Eq. 3.

\section{One-Shot Entanglement Cost of Preparing $\rho$}\label{One_Shot_Proof}
When talking about the one-shot entanglement cost, the most commonly used quantity is the Schmidt number \cite{Buscemi11,BranD2011,Berta+2013}. To explain why this is, let us look at the optimal preparation protocol. One wishes to prepare a mixed state $\rho$ from a maximally entangled state $\ket{\Phi^+_d}$. One picks a decomposition of $\rho = \sum_{i} p_i \ket{\phi_i}\bra{\phi_i}$, and samples from the distribution $\{p_i\}$. If outcome $i$ is obtained, one determinstically transforms $\ket{\Phi^+_d}$ into $\ket{\phi_i}$ via LOCC. Thus the expected state at the end of the process is:
\begin{equation}
\sum_i p_i \ket{i}\bra{i} \otimes \ket{\phi_i}\bra{\phi_i}
\end{equation}
which reproduces $\rho$ once the classical information, $i$, is ``forgotten".

In order to be able to prepare $\ket{\phi_i}$ without error, one must have that $d \geq \mathrm{S}_R (\phi_i)$ \cite{N1999}. By optimising over all decompositions of $\rho$, and counting maximally entangled qubits ($d=2$) as the resource unit, it is easy to arrive at the Schmidt number, 
\begin{equation}
S_N(\rho):=\min_{\{p_i,\phi_i\}} \max_{i} \log S_R(\phi_i).
\end{equation}
For us, this method led to a natural question: what if, when you receive outcome $i$, the state $\phi_i$ has a lower Schmidt number than the maximum? One could prepare the required state using $\ket{\Phi^+_{d'}}$ with $d'<d$, without sacrificing any quality of the final preparation. In order to exploit this saving, we must allow the preparer to begin with a state of the form:
\begin{equation}\label{mixedresource}
\omega = \sum_{m} p_m \ket{m}\bra{m}\otimes \ket{\Phi^+_m}\bra{\Phi^+_m}.
\end{equation}
i.e. a labelled mixture of maximally entangled states, of various dimensions. If $\rho$ is preperable from such a state, then we would define the corresponding entanglement cost as $\sum_{m} p_m \log m$ - the \emph{expected} number of qubit maximally entangled pairs consumed.
\begin{lemma}
With this additional freedom in the resource state, the one-shot entanglement cost is given by the Schmidt measure of $\rho$.
\end{lemma}
\begin{proof}
We first show that the Schmidt measure is an achievable rate. Take the optimal decomposition $\rho=\sum p_i \ket{\phi_i}\bra{\phi_i}$. By Nielsen's theorem \cite{N1999}, there exists a LOCC channel $\mathcal{E}_i$ such that $\mathcal{E}_i(\ket{\Phi^+_{d_i}}\bra{\Phi^+_{d_i}}) = \ket{\phi_i}\bra{\phi_i}$, where $d_i = S_{R}(\phi_i)$.\\ 

We now define the channel
\begin{equation}
\mathcal{E}\left(\ket{i}\bra{j} \otimes \rho\right)=\begin{cases}
\mathcal{E}_{i}(\rho) &\text{ if } i = j,\\
0 & \text{ else.}\\
\end{cases}
\end{equation}
Since each $\mathcal{E}_{i}$ is a LOCC channel, and $i$ is a classical label which can freely communicated, $\mathcal{E}$ is clearly also a LOCC channel. Applying this channel to the state $\omega = \sum_{i} p_i \ket{\Phi^+_{d_i}}\bra{\Phi^+_{d_i}}$ we obtain the state $\rho$, with entanglement cost $\sum_{i} p_i \log d_i = S_{M}(\rho)$.\\

In order to proof the converse, we will investiage the necessary conditions for exact state creation using the fidelity. We follow closely the proof used in \cite{BD2011}, which was then used in \cite{Buscemi11} to prove rigorously the Schmidt number result discussed above.\\

We start by using a result by Lo and Popescu \cite{LP1999}, which states that, starting from a maximally entangled state, any LOCC preparable state can be expressed in the form:
\begin{equation}
\sum_{k} \left(K_{k}\otimes U_{k}\right)\ket{\Phi^+_m}\bra{\Phi^+_m}\left(K_{k}\otimes U_{k}\right)^{\dagger},
\end{equation}
where $K_k$ are Kraus operators and $U_k$ unitaries.
Since we begin with a classically labelled set of pure states, each of which has an LOCC channel applied to it, any state which can be prepared from resource \ref{mixedresource} is expressible as:
\begin{equation}
\Lambda(\omega)=\sum_{m,k} q_m\left(K^m_{k}\otimes U^{m}_{k}\right)\ket{\Phi^+_m}\bra{\Phi^+_m}\left(K^m_{k}\otimes U^m_{k}\right)^{\dagger}.
\end{equation}
We now are interested in the fidelity between $\Lambda(\omega)$ and the target state $\rho$, which by Uhlmann's theorem \cite{U1976} is equal to $\max | \bra{\psi_{\rho}}\ket{\psi_{\Lambda(\omega)}}|^2$, where $\psi_{\Lambda(\omega)}$ is an explicit purification of $\Lambda(\omega)$, and the maximisation is taken over all purifications $\ket{\psi_{\rho}}$ of $\rho$. We will choose 
\begin{equation}
\ket{\psi_{\Lambda(\omega)}} = \sum_{m,k} \ket{m,k} \otimes \sqrt{q_m} \left(K^m_{k}\otimes U^{m}_{k}\right)\ket{\Phi^+_m}
\end{equation}
as our purification of $\Lambda(\omega)$, whilst from \cite{HJW1993} any purification of $\rho$ is expressible as $\sum_{i} \sqrt{p_i} \ket{i}\ket{\phi_i}$, where $\rho=\sum_i p_i \ket{\phi_i}\bra{\phi_i}$.\\

From this, we see that our fidelity is
\begin{equation}
\max |\sum_{i}\bra{\phi_i}\left(\sqrt{q_{m(i)}} \left(K^{m(i)}_{k(i)}\otimes U^{m(i)}_{k(i)}\right)\ket{\Phi^+_{m(i)}}\right)
\end{equation}
where $m(i),k(i)$ signify the one-to-one relation between $i$ and pairs $(m,k)$.\\

Before we continue, we will rewrite $\left(\mathbb{I}\otimes U^{m(i)}_{k(i)}\right)\ket{\Phi^+_{m(i)}}$. (We suppress the $i$ for convenience).\\
\begin{align*}
 \left(\mathbb{I}\otimes U^{m}_{k}\right)\ket{\Phi^+_{m}} &=    \frac{1}{\sqrt{m}}\sum^{m}_{r}\ket{r}U^{m}_{k}\ket{r}\\
 &=\frac{1}{\sqrt{m}}P_{m}\sum_{r}^d W^{i}_{k,1}\ket{e_r}U^{m}_{k}W^{i}_{k,2}\ket{f_r}\\
 &= \frac{1}{\sqrt{m}}P_{m}\sum_{r}^d V^{i}_{k}\ket{e_r}\ket{f_r}\\
\end{align*}
where $\ket{\phi_i}=\sum_{r}\sqrt{\lambda_r}\ket{e_r^i}\ket{f_r^i}$ is the Schmidt decomposition of $\ket{\phi_i}$, $W^{i}_{k,1}\ket{e_r} = \ket{r}$, $W^{i}_{k,2}\ket{f_r} = \ket{r}$ , $P_{m}=\sum_{j}^m \ket{j}\bra{j}$, and $V^{i}_{k} = (U^{m}_{k}W^{i}_{k,2})^T W^{i}_{k,1}$, where transposition is in the computational basis.  $d$ is the full dimension of the first system.\\
Putting this rearrangement into our fidelity, we see that the we can rewrite an individual term as:
\begin{equation}
\bra{\phi_i}\left(\sqrt{q_{m(i)}} \left(K^{m(i)}_{k(i)}\otimes U^{m(i)}_{k(i)}\right)\ket{\Phi^+_{m(i)}}\right) = \mathrm{Tr}\left(\frac{\sqrt{q_{m(i)}}}{\sqrt{m(i)}}\sqrt{\rho^i}K^{m(i)}_{k(i)}P_{m(i)}V_{k(i}^{m(i)}\right).
\end{equation}
This means our fidelity becomes:
\begin{align}
F&= \left|\sum_{i}\sqrt{q_{m(i)}}\sqrt{p_i} \mathrm{Tr}\left(\frac{1}{\sqrt{m(i)}}\sqrt{\rho^i}K^{m(i)}_{k(i)}P_{m(i)}V_{k(i}^{m(i)}\right)\right|^2 \nonumber \\
&\leq \left(\sum_{i} \left|\mathrm{Tr}\left(\frac{1}{\sqrt{m(i)}}\sqrt{\rho^i}K^{m(i)}_{k(i)}P_{m(i)}V_{k(i}^{m(i)}\right)\right|\right)^2 \nonumber\\
&\leq \left(\sum_{i} \left(\frac{q_{m(i)}}{m(i)}\mathrm{Tr}\left[P_{m(i)}K^{m(i)\,\dagger}_{k(i)}K^{m(i)}_{k(i)}\right]\mathrm{Tr}\left[p_i\rho_i P_{m(i)}\right]\right)^{1/2}\right)^{2}, \label{2ndlastbound}
\end{align}
where we have used the triangle inequality followed by the Cauchy Schwarz inequality. 
Finally, we can note that 
\begin{align*}
\sum_{i} \frac{q_{m(i)}}{m(i)}\mathrm{Tr}\left[P_{m(i)}K^{m(i)\,\dagger}_{k(i)}K^{m(i)}_{k(i)}\right] &= \sum_{m}\frac{q_{m}}{m}\mathrm{Tr}\left[P_{m}\left(\sum_{k}K^{m\,\dagger}_{k} K^{m}_{k}\right)\right] \\
&= \sum_m\frac{q_{m}}{m}\mathrm{Tr}\left(P_{m}\right) = 1.
\end{align*}
Using the fact that $\sum_{i}\sqrt{p_iq_i}\leq \sqrt{\sum_i q_i}$ when $\sum_{i}p_i =1$, then our fidelity bound simplifies further to:
\begin{equation}
F \leq \left(\left(\sum_{i} p_{i}\mathrm{Tr}\left(\rho_i P_{m(i)}\right)\right)^{1/2}\right)^{2}
\end{equation}
and the two powers cancel.\\

Since we are demanding that $\rho$ can be created be perfect fidelity, we must choose an initial $\omega$ which gives an upper bound of 1. We can see that this forces $P_{m(i)} \geq \mathrm{rank}(\rho_i) = S_R(\ket{\phi_i})$. With this knowledge, we can return to our previous upper bound, Eq. \ref{2ndlastbound}. The maximum of $\sum_{i}\sqrt{p_iq_i}$ over a distribution $p_i$, where $q_i$ is \emph{also} given by a probability distribution, is 1, achieved when $p_i = q_i$. This tells us that, in order for fidelity to be 1, is is \emph{necessary} that $$\frac{q_{m(i)}}{m(i)}\mathrm{Tr}\left[P_{m(i)}\left(K^{m(i)\,\dagger}_{k(i)} K^{m(i)}_{k(i)}\right)\right]= p_i.$$ This implies that:
$q_m = \sum_{i| m(i)=m} p_i$.\\

Finally, let us consider the smallest entanglement cost that satisfies this necessary condition. We wish to minimize:
\begin{equation}
\min\{\sum_{m} \left(\sum_{i, m(i)=m} p_i\right) \log m  \mid m(i) \geq S_R(\ket{\phi_i})\}.
\end{equation}
For a given distribution, this is done by assigning $m(i)= S_R(\ket{\phi_i})$. Our minimum is therefore given by $\sum_i p_i S_R(\ket{\phi_i})$ minimised over all decomposition of $\rho$ - \emph{exactly} the definition of the Schmidt number! Thus the Schmidt number is a lower bound of the expected entanglement cost, and an achievable value - proving the result.
\end{proof}

\end{document}